\documentclass{article}

\usepackage{algpseudocode}
\usepackage{algorithm}
\usepackage[utf8]{inputenc}
\usepackage{graphicx}
\usepackage{listings}
\usepackage{float}
\usepackage{longtable}
\usepackage{xcolor}
 \usepackage{placeins}
\usepackage{appendix}
\usepackage{amsfonts}
\usepackage{amsmath}
\usepackage[english]{babel}
\usepackage{mdframed}
\usepackage[a4paper]{geometry} \usepackage[hidelinks]{hyperref}
\usepackage{wrapfig} \usepackage{paralist} \usepackage{amsthm}
\usepackage{tikz}

\usepackage{cleveref} \crefname{figure}{Figure}{Figures}
\Crefname{figure}{Figure}{Figures}
\crefname{section}{Section}{Sections}
\Crefname{section}{Section}{Sections}
\crefname{listing}{Listing}{Listings}
\Crefname{listing}{Listing}{Listings}

\newtheorem{definition}{Definition}
\newtheorem{lemma}{Lemma}
\newtheorem{corollary}{Corollary}
\newtheorem{theorem}{Theorem}

\lstdefinelanguage{quint}{
  morekeywords={action, run, bool, all, any, and, or, val, match},
  sensitive=true,                             morecomment=[l]{//},                        morecomment=[s]{/*}{*/},                    morestring=[b]",                            morecomment=[l]{//},
} 

\newcommand{\Corr}[0]{\textit{Correct}}
\newcommand{\Faulty}[0]{\textit{Faulty}}
\newcommand{\tqc}[0]{\mathbb{T}}
\newcommand{\cqc}[0]{\mathbb{C}}
\newcommand{\signers}[1]{\textit{Signers}(#1)}
\newcommand{\eqmod}{\approx}
\newcommand{\leader}[1]{\textit{leader}(#1)}
\newcommand{\delay}{\Delta_{\text{rcv}}}
\newcommand{\timeout}{\Delta_{\text{timeout}}}
\newcommand{\gst}{\text{GST}}
\newcommand{\tsync}{\Delta_{\text{sync}}}
\newcommand{\rcv}[1]{\searrow{}{#1}}
\newcommand{\snd}[1]{{#1}\nearrow{}}
\newcommand{\msg}[1]{\langle{#1}\rangle}

\newcommand{\invOneName}{\textsc{LocalHighVote}}
\newcommand{\invTwoName}{\textsc{HighVoteInTimeoutQC}}
\newcommand{\invOne}[1]{\invOneName({#1})}
\newcommand{\invTwo}[1]{\invTwoName({#1})}
\newcommand{\tlap}{$\textsc{TLA}^{+}$}
\newcommand{\NatZero}{\mathbb{N}_0} 
\title{ChonkyBFT: Consensus Protocol of ZKsync}

\author{
    Bruno França\textsuperscript{1}, 
    Denis Kolegov\textsuperscript{1}, 
    Igor Konnov\textsuperscript{2},  
    and
    Grzegorz Prusak\textsuperscript{1} 
}

\date{
    \textsuperscript{1}Matter Labs \\
    \textsuperscript{2}Independent formal methods researcher
}

\begin{document}

\maketitle

\begin{abstract}

We present ChonkyBFT, a partially-synchronous Byzantine fault-tolerant (BFT) consensus protocol used in the ZKsync system. The proposed protocol is a hybrid protocol inspired by FAB Paxos, Fast-HotStuff, and HotStuff-2. It is a committee-based protocol with only one round of voting, single slot finality, quadratic communication, and $n \ge 5f+1$ fault tolerance. This design enables its effective application within the context of the ZKsync rollup, achieving its most critical goals: simplicity, low transaction latency, and reduced system complexity. The target audience for this paper is the ZKsync community and others worldwide who seek assurance in the safety and security of the ZKsync protocols. The described consensus protocol has been implemented, analyzed, and tested using formal methods.

\end{abstract}

\section{Introduction}

In recent years, research on consensus algorithms has often focused on theoretical optimizations that may not align with the practical challenges of implementing these algorithms. This has led to researchers prioritizing metrics that are less impactful in real-world deployments. For instance, tables comparing algorithms frequently highlight metrics that, while theoretically appealing, offer limited practical significance. Below, we outline commonly emphasized metrics that may not adequately reflect real-world performance considerations.
\begin{itemize}

    \item \textbf{Authenticator complexity}. Optimizing to have fewer signatures was a priority decades ago when crypto operations were expensive. Today, digital signatures are fast and small. However, many papers still report this measure and even go as far as suggesting threshold signatures instead of multisignatures, which introduces a much more complex step of distributed key generation instead of spending some more milliseconds on verifying the signatures.

    \item \textbf{Message complexity}. In theory, reducing the number of messages exchanged across the network should improve the algorithm's performance. However, in practice, the actual performance gain depends on where the system bottleneck lies. For instance, even with linear communication, if the leader must still send and receive~$N$ messages, the improvement may be negligible. Moreover, this approach often oversimplifies the cost of messages, treating all messages equally. In reality, a block proposal can be several megabytes in size, whereas a block commit may only be a few kilobytes.

    \item \textbf{Block latency}. Block times of, for example, 0.1 seconds are not necessarily a reliable performance indicator if finality requires waiting for 100 blocks. The key metric that matters is the time it takes for a user to see their transaction finalized. Some algorithms claim to achieve only one round of voting; however, they introduce an additional round within the block broadcast mechanism. Or pipeline the voting rounds, thus requiring several blocks to finalize. These approaches can result in higher user-perceived latency, despite achieving shorter block times.
\end{itemize}

While the above metrics are often highlighted, other aspects of consensus protocol design have a far greater impact in practical settings. Below, we discuss the metrics that we believe are most relevant for real-world consensus algorithm performance.

\begin{itemize}
    \item \textbf{Systemic complexity}. This aligns with the discussion on systemic versus encapsulated complexity \cite{Buterin2022Complexity}. Consensus algorithms do not operate in isolation; they are designed to support a wide range of applications. A clear example of this issue is the distinction between probabilistic and provable finality. Algorithms that provide probabilistic finality introduce additional complexity for applications, as exchanges, multi-chain dApps, hybrid dApps, block explorers, wallets, and similar systems must determine an appropriate number of confirmations for each chain. In contrast, algorithms with provable finality deliver a clear and deterministic signal that applications can rely on. This distinction is significant enough that even Ethereum is planning to adopt single-slot finality.

    \item \textbf{Simplicity}. When designing and implementing an algorithm, it is essential to consider the trade-off between complexity and practicality. While saving one round-trip in an optimistic scenario may seem advantageous, it may not be worthwhile if the algorithm is too complex to formalize and model effectively. Additionally, a complex implementation that requires significant resources—such as multiple engineers and extensive audits—can introduce further challenges. Simple algorithms that are straightforward to implement and can be formally proven offer greater security and reliability. A bug that results in downtime, or worse, safety violations, has a far more detrimental impact on user experience than marginally slower block times.

    \item \textbf{Transaction latency}. As previously discussed, the only latency that truly matters is the one experienced by the user.

\end{itemize}

\textbf{Lessons learned.}
For our specific use case, we have gained several insights from researching and implementing previous consensus algorithms:
\begin{enumerate}

    \item Chained consensus is not beneficial for our use case. It does not improve throughput or latency but adds unnecessary systemic complexity. Instead, we finalize every block.
    
    \item Reducing fault tolerance can minimize the number of voting rounds, as demonstrated by FaB Paxos. By decreasing our fault tolerance requirement from $3f+1$ to $5f+1$, we are able to finalize consensus in a single voting round \cite{2roundBFT}.
    
    \item Linear communication is not always advantageous. Quadratic communication among replicas simplifies security, as there are fewer cases where the impact of a malicious leader needs to be considered. It also simplifies implementation by allowing a clear separation of the leader component, as well as view changes, where constant timeouts suffice. For example, Jolteon/Ditto \cite{Gelashvili2021JolteonDitto} adopted this approach after encountering challenges while implementing HotStuff. Additionally, the performance trade-off is likely minimal, as demonstrated by ParBFT \cite{Dai2023ParBFT}.
    
    \item Re-proposals can be used to ensure that no "rogue" blocks exist—a challenge that has received little attention so far (to the best of our knowledge) and is particularly relevant to public blockchains. In committee-based consensus algorithms, it is possible for a commit quorum certificate (to use HotStuff's terminology) to be formed without being received by enough replicas. This can result in a timeout and the proposal of a new block. Most algorithms address this by declaring the old block invalid. While all honest replicas will agree on which block is canonical, a participant who only sees the old block without knowledge of the timeout may mistakenly believe that it was finalized. This undermines the desirable property of verifying a block's inclusion in the chain by inspecting only the block itself, without needing the entire chain. To address this, we require that block proposals following a timeout (where a commit quorum certificate might have been formed) re-propose the previous block. This approach guarantees that any block with a valid commit quorum certificate is part of the chain—it may not have been finalized in the current view, but it was certainly finalized.
    
    \item Always justify messages to remove time dependencies, a lesson we learned from Fast-HotStuff. Messages should carry sufficient information to allow any replica to independently verify their validity without relying on additional data (except for previous blocks, which are external to the consensus algorithm). Failing to do so can introduce subtle timing dependencies. For example, Tendermint contained a bug—discovered years later—where the leader had to wait for the maximum network delay at the end of each round to avoid a deadlock. If this wait did not occur, the system could lock. Interestingly, HotStuff-2 reintroduces this timing dependency to eliminate one round-trip, which significantly increases the complexity of modeling and implementing the system.
    
    \item Make garbage collection and reconfiguration integral parts of the algorithm. These components will inevitably need to be implemented, and failing to specify and model them upfront may result in awkward and inefficient implementations later.
    
\end{enumerate}

FaB Paxos satisfies the first 4 points, and Fast-HotStuff satisfies the 5th. ChonkyBFT is basically FaB Paxos with some ideas from Fast-HotStuff/HotStuff-2.

\section{ChonkyBFT Overview}
In this section, we provide an intuitive overview of our consensus protocol to help readers grasp the fundamental concepts before delving into the formal specification. As in many consensus protocols for blockchains, the goal
of ChonkyBFT is to extend the latest known block of the blockchain with a set of new transactions, which
collectively form a new block. Whenever correct replicas manage to do so, we say that they commit the new block.

The system comprises $n \ge 5f+1$ replicas, where $f$ is the maximum number of faulty replicas to be
tolerated.

The consensus protocol operates in successive \textit{views}, each of them 
having an algorithmically-defined leader called~$leader(v)$ for each view~$v$,
which can round-robin in the simplest case, but it does not have to be
round-robin. In each view, the leader proposes a single block, and the other
replicas attempt to commit it. The protocol must account for three primary
cases:

\begin{enumerate}

\item \textbf{Successful commit in a view:} The proposal is committed by all correct replicas,
  and the protocol progresses smoothly.

\item \textbf{Partial commit with timeout in view:} Some correct replicas are timed out, while other replicas have committed the proposal. 

\item \textbf{Timeout without commit in view:} All correct replicas time out, and the proposal is not committed by any replica. 

\end{enumerate}

The goal is to optimize for the first case, enabling the protocol to operate at maximum efficiency during normal conditions within the synchrony model. For the second case, mechanisms are in place to ensure that the next leader re-proposes the block to achieve eventual agreement among all correct replicas. In the third case, the protocol may still re-propose the missed block to maintain safety, as it cannot distinguish between cases two and three without further information.

Each view consists of two or three phases: the \textit{prepare} phase, the \textit{commit} phase, and optional 
\textit{timeout} phase. A new view starts either when the previous proposal is successfully committed or when replicas experience a timeout. 

\subsection{Protocol Mechanics: The Happy Path}

The leader broadcasts a \textit{Proposal} message containing the proposed \textit{block} and a \textit{justification} proving that the leader is authorized to propose in the current view. 

Upon receiving the \textit{Proposal} message, each replica verifies the proposal. If valid, the replica broadcasts a \textit{SignedCommitVote} message, which is a signature over the replica's \textit{CommitVote} message, which includes the proposal's view and also the block number and hash, inferred from the received justification.

If a replica collects $n-f$ such commit votes, it constructs a \textit{quorum certificate} denoted by 
\textit{CommitQC}. This certificate proves the block can be safely committed, even if the leader fails to disseminate 
it. The replica then commits the block, broadcasts to all replicas a \textit{NewView} message containing the 
\textit{CommitQC}, and advances to the next view. The new leader uses the constructed or received commit quorum certificate as a justification for entering the new view and proving its privilege to propose in this view.

Replicas can commit to a block through either of the following:

\begin{itemize}
\item Collecting $n-f$ votes sufficient for creating a \textit{CommitQC} at the end of a view.

\item Receiving \textit{CommitQC} indirectly in a \textit{Proposal}, \textit{NewView} or \textit{Timeout} message.
\end{itemize}

Figure~\ref{fig:happy1} illustrates the simplest happy path scenario for~$n=6$ and~$f=1$. In this scenario,
replica~0 is the leader in view~1. It broadcasts its proposal for a new block. We call this proposal~$P_1$
and use the notation~$\snd{P_1}$ for broadcasting the proposal~$P_1$. As soon as the replicas from~0 to~4
receive the proposal~$P_1$, they broadcast their commit votes~$C_0, \dots, C_4$ to all replicas, respectively.
Finally, the replicas from~0 to~4 receive the votes from one another, which we write as $\rcv{C_{0..4}}$.
Then, they collect their commit quorums
of $n - f = 5$ votes, commit the block, and send their new-view messages, denoted with~$\snd{N_i}$.
After that, they all switch to view~2. Note that the faulty replica~5 is not involved in this scenario at all:
It simply ignores all messages and does not send anything.

In the above example, all correct replicas had to collect five commit votes and issue their commit quorum
certificates. Figure~\ref{fig:happy2} shows an optimal scenario. Once replica~1 constructs its
quorum of five commit votes, it sends its new-view message~$N_1$. The other correct replicas~$0, 2, 3, 4$
receive the message~$N_1$, verify the commit quorum certificate, immediately commit the block, and switch
to view~2. Now the leader for view-2 can use the \texttt{CommitQC} as justification to propose a new block.

\begin{figure}

\begin{tikzpicture}[
    message/.style={->, >=latex, thick},
    phase/.style={dotted, thick},
    process/.style={draw, dashed},
    tag/.style={fill=gray!20!white,rounded corners=4pt},
]

\foreach \i/\name in {0/0, 1/1, 2/2, 3/3, 4/4, 5/5} {
    \node[anchor=west] (p\i) at (0, -\i) {\name};
    \draw[process] (1, -\i) -- (12, -\i);
}

\foreach \x/\name in {1/prepare (v=1), 6/, 11/} {
    \draw[phase] (\x, 0.5) -- (\x, -5);
}
\foreach \x/\name in {3.25/prepare (v=1), 8.5/commit (v=1), 12.5/prepare (v=2)} {
    \node[above] at (\x, 1) {\name};
}

\node[tag] at (2.2, .4) { $\rcv{P_1}; \snd{C_0}$ };

\draw[message,color=green!50!black] (1, 0) -- (3.5, -1);
\draw[message,color=green!50!black] (1, 0) -- (3, -2);
\draw[message,color=green!50!black] (1, 0) -- (3, -3);
\draw[message,color=green!50!black] (1, 0) -- (2.5, -4);
\draw[message,color=green!50!black] (1, 0) -- (2, -5);

\node[tag] at (4.2, -.5) { $\rcv{P_1};\snd{C_1}$ };
\node[tag] at (3.8, -1.5) { $\rcv{P_1};\snd{C_2}$ };
\node[tag] at (3.9, -2.5) { $\rcv{P_1};\snd{C_3}$ };
\node[tag] at (3.7, -3.5) { $\rcv{P_1};\snd{C_4}$ };

\node[tag] at (8.3, .4) { $\rcv{C_{0..4}};\snd{N_0}$ };
\node[tag] at (7.7, -.5) { $\rcv{C_{0..4}};\snd{N_1}$ };
\node[tag] at (8.7, -1.5) { $\rcv{C_{0..4}};\snd{N_2}$ };
\node[tag] at (7.3, -2.5) { $\rcv{C_{0..4}};\snd{N_3}$ };
\node[tag] at (9.3, -3.5) { $\rcv{C_{0..4}};\snd{N_4}$ };

\end{tikzpicture}

\caption{The simplest happy path scenario: Replica 0 broadcasts proposal $P_1$; replicas $0, \dots, 4$
 receive $P_1$ and broadcast commit votes $C_i$, for $i =0, \dots, 4$; replicas $0, \dots, 4$
 receives~5 commit votes $C_0, \dots, C_4$, commits the block and send their new-view votes~$N_i$, for $i =0, \dots, 4$. Replica 5 is faulty and ignores all messages.}\label{fig:happy1}

\end{figure}
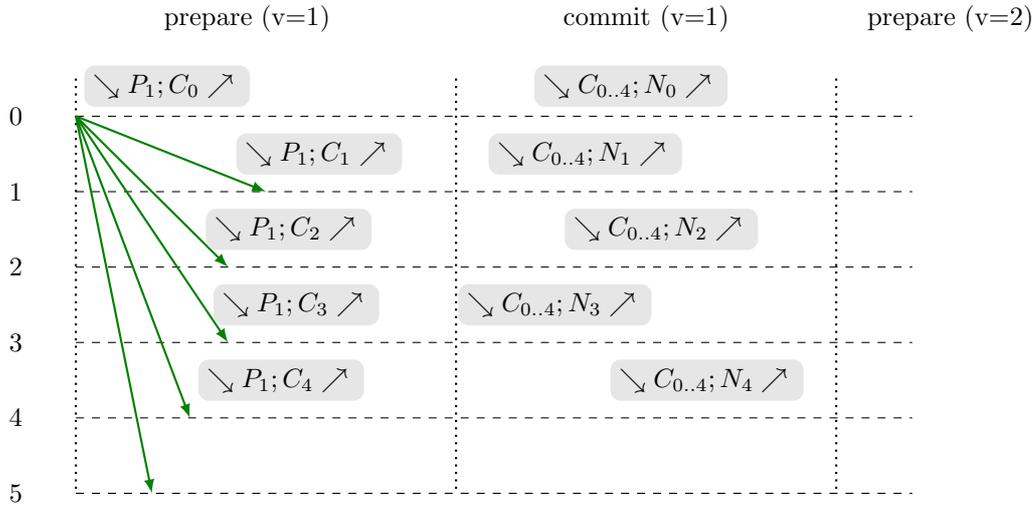 \begin{figure}

\begin{tikzpicture}[
    message/.style={->, >=latex, thick},
    phase/.style={dotted, thick},
    process/.style={draw, dashed},
    tag/.style={fill=gray!20!white,rounded corners=4pt},
]

\foreach \i/\name in {0/0, 1/1, 2/2, 3/3, 4/4, 5/5} {
    \node[anchor=west] (p\i) at (0, -\i) {\name};
    \draw[process] (1, -\i) -- (12, -\i);
}

\foreach \x/\name in {1/, 6/, 11/} {
    \draw[phase] (\x, 0.5) -- (\x, -5);
}
\foreach \x/\name in {3.25/prepare (v=1), 8.5/commit (v=1), 12.5/prepare (v=2)} {
    \node[above] at (\x, 1) {\name};
}

\node[tag] at (2.2, .4) { $\snd{P_1}; \snd{C_0}$ };

\draw[message,color=green!50!black] (1, 0) -- (3.5, -1);
\draw[message,color=green!50!black] (1, 0) -- (3, -2);
\draw[message,color=green!50!black] (1, 0) -- (3, -3);
\draw[message,color=green!50!black] (1, 0) -- (2.5, -4);
\draw[message,color=green!50!black] (1, 0) -- (2, -5);

\node[tag] at (4.4, -.5) { $\rcv{P_1};\snd{C_1}$ };
\node[tag] at (3.9, -1.5) { $\rcv{P_1};\snd{C_2}$ };
\node[tag] at (3.9, -2.5) { $\rcv{P_1};\snd{C_3}$ };
\node[tag] at (3.6, -3.5) { $\rcv{P_1};\snd{C_4}$ };

\node[tag] at (8.5, .5) { $\rcv{N_1}$ };
\node[tag] at (7, -.5) { $\rcv{C_{0..4}};\snd{N_1}$ };
\node[tag] at (8.5, -1.5) { $\rcv{N_1}$ };
\node[tag] at (9.5, -2.5) { $\rcv{N_1}$ };
\node[tag] at (9, -3.5) { $\rcv{N_1}$ };

\end{tikzpicture}

\caption{An optimized happy path scenario: Replica 0 broadcasts proposal $P_1$; replicas $0, \dots, 4$
 receive $P_1$ and broadcast commit votes $C_i$, for $i =0, \dots, 4$; replica 1
 receives~5 commit votes $C_0, \dots, C_4$, commits the block and sends its new-view vote~$N_1$; replicas $0, 2, 3, 4$
 receive~$N_1$, immediately commit the block, and proceed with the next view. Replica 5 is faulty and ignores all messages.}\label{fig:happy2}

\end{figure}
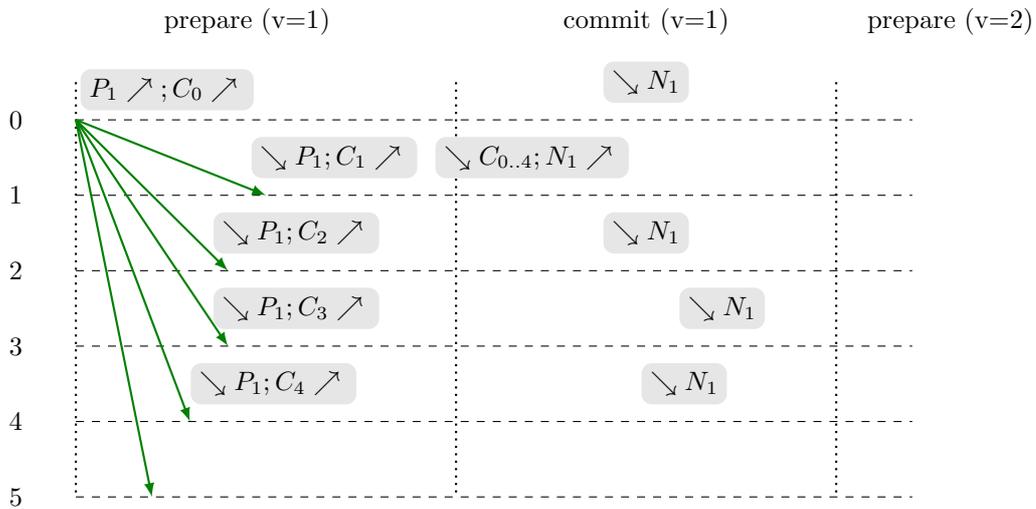 
\subsection{Handling Unhappy Paths: Timeouts and Re-proposals}

In adverse conditions, such as network delays or malicious leaders, the protocol must handle scenarios where progress
is impeded. Replicas wait for a predefined duration~$\timeout$ before timing out.

Upon timing out, replicas broadcast \textit{SignedTimeoutVote} message, which 
is a signature over the replica's \textit{TimeoutVote} message and the highest 
commit QC observed by the replica. The \textit{TimeoutVote} message includes 
the replica's current view, the commit vote with the highest view that this 
replica has signed so far, and the view number of the highest commit quorum 
certificate that this replica has observed.

When replicas collect $n-f$ timeout votes, they assemble a \textit{timeout certificate} denoted by \textit{TimeoutQC}. The new leader uses the TimeoutQC certificate to justify a new block proposal or block re-proposal in the next view. 

\begin{figure}

\begin{tikzpicture}[
    message/.style={->, >=latex, thick},
    phase/.style={dotted, thick},
    process/.style={draw, dashed},
    tag/.style={fill=gray!20!white,rounded corners=4pt},
]

\foreach \i/\name in {0/0, 1/1, 2/2, 3/3, 4/4, 5/5} {
    \node[anchor=west] (p\i) at (0, -\i) {\name};
    \draw[process] (1, -\i) -- (14.5, -\i);
}

\foreach \x/\name in {1/, 6/, 8.5/, 12/, 14.5/} {
    \draw[phase] (\x, 0.5) -- (\x, -5);
}
\foreach \x/\name in {3.75/prepare (v=1), 7.3/commit (v=1), 10.25/timeout (v=1), 13.25/prepare (v=2)} {
    \node[above] at (\x, 1) {\name};
}

\draw[message,color=green!50!black] (1, 0) -- (3.5, -1);
\draw[message,color=green!50!black] (1, 0) -- (3, -2);
\draw[message,color=green!50!black] (1, 0) -- (3, -3);
\draw[message,color=green!50!black] (1, 0) -- (2.5, -4);
\draw[message,color=green!50!black] (1, 0) -- (2, -5);

\node[tag] at (2, .3) { $\snd{P_1}; \snd{C_0}$ };
\node[tag] at (4.5, -.3) { $\rcv{P_1};\snd{C_1}$ };
\node[tag] at (3.4, -1.4) { $\rcv{P_1};\snd{C_2}$ };
\node[tag] at (3.2, -2.3) { $\rcv{P_1};\snd{C_3}$ };
\node[tag] at (4.4, -3.6) { $\rcv{P_1};\snd{C_4}$ };

\node[tag] at (9.6, .4) { $\textit{timeout}; \snd{T_0}$ };
\node[tag] at (9.7, -.6) { $\textit{timeout}; \snd{T_1}$ };
\node[tag] at (9.6, -1.6) { $\textit{timeout}; \snd{T_2}$ };
\node[tag] at (9.7, -2.6) { $\textit{timeout}; \snd{T_3}$ };
\node[tag] at (9.6, -3.6) { $\textit{timeout}; \snd{T_4}$ };

\node[tag] at (11.5,  0.4) { $\rcv{T_{0..4}}$ };
\node[tag] at (11.5, -0.6) { $\rcv{T_{0..4}}$ };
\node[tag] at (11.5, -1.6) { $\rcv{T_{0..4}}$ };
\node[tag] at (11.5, -2.6) { $\rcv{T_{0..4}}$ };
\node[tag] at (11.5, -3.6) { $\rcv{T_{0..4}}$ };
\draw[message,color=red!50!black] (11.5, 0) -- (12, 0);
\draw[message,color=red!50!black] (11.5, -1) -- (12, -1);
\draw[message,color=red!50!black] (11.5, -2) -- (12, -2);
\draw[message,color=red!50!black] (11.5, -3) -- (12, -3);
\draw[message,color=red!50!black] (11.5, -4) -- (12, -4);

\draw[message,color=green!50!black] (13.75, -3) -- (14.5, -0);
\draw[message,color=green!50!black] (13.75, -3) -- (14.5, -1);
\draw[message,color=green!50!black] (13.75, -3) -- (14.5, -2);
\draw[message,color=green!50!black] (13.75, -3) -- (14.5, -4);

\node[tag] at (13.5, -3.4) { $\snd{P_1}; \snd{C_5}$ };

\end{tikzpicture}

\caption{An unhappy path scenario: A correct proposer sends a proposal,
  commit votes are lost owing to asynchrony, all other
  correct replicas time out.}\label{fig:unhappy1}

\end{figure}
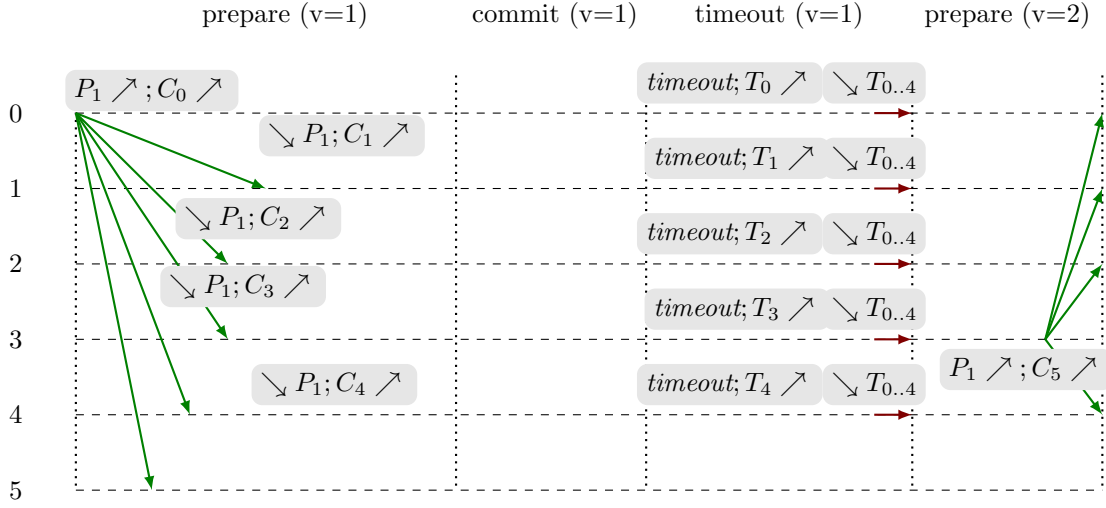 \begin{figure}

\begin{tikzpicture}[
    message/.style={->, >=latex, thick},
    phase/.style={dotted, thick},
    process/.style={draw, dashed},
    tag/.style={fill=gray!20!white,rounded corners=4pt},
    taggreen/.style={fill=green!20!white,rounded corners=4pt},
    tagred/.style={fill=red!20!white,rounded corners=4pt},
]

\foreach \i/\name in {0/0, 1/1, 2/2, 3/3, 4/4, 5/5} {
    \node[anchor=west] (p\i) at (0, -\i) {\name};
    \draw[process] (1, -\i) -- (14.5, -\i);
}

\foreach \x/\name in {1/, 5.1/, 7.3/, 12/, 14.5/} {
    \draw[phase] (\x, 0.5) -- (\x, -5);
}
\foreach \x/\name in {3/prepare (v=1), 6.1/commit (v=1), 10.25/timeout (v=1), 13.25/prepare (v=2)} {
    \node[above] at (\x, 1) {\name};
}

\draw[message,color=green!50!black] (1, -5) -- (2, 0);
\draw[message,color=green!50!black] (1, -5) -- (2.5, -1);
\draw[message,color=green!50!black] (1, -5) -- (3, -2);

\draw[message,color=orange!50!black] (1, -5) -- (3, -3);
\draw[message,color=orange!50!black] (1, -5) -- (2.5, -4);

\node[taggreen] at (2, .3) { $\snd{P_1}; \snd{C_0}$ };
\node[taggreen] at (4, -.4) { $\rcv{P_1};\snd{C_1}$ };
\node[taggreen] at (4, -1.6) { $\rcv{P_1};\snd{C_2}$ };

\node[tagred] at (4, -2.6) { $\rcv{P_2};\snd{C_3}$ };
\node[tagred] at (3.7, -3.6) { $\rcv{P_2};\snd{C_4}$ };

\node[taggreen] at (8.5, .4) { $\textit{timeout}; \snd{T_0}$ };
\node[taggreen] at (8.5, -.6) { $\textit{timeout}; \snd{T_1}$ };
\node[taggreen] at (8.5, -1.6) { $\textit{timeout}; \snd{T_2}$ };
\node[tagred] at (8.5, -2.6) { $\textit{timeout}; \snd{T_3}$ };
\node[tagred] at (8.5, -3.6) { $\textit{timeout}; \snd{T_4}$ };
\node[tagred] at (8.5, -4.3) { $\textit{timeout}; \snd{T_5}$ };

\node[tag] at (10, -4.9) { $\rcv{T_{0\dots5}}; \snd{N_5}$ };

\node[tag] at (11.4, .4) { $\rcv{N_5}$ };
\node[tag] at (11.4, -.6) { $\rcv{N_5}$ };
\node[tag] at (11.4, -1.6) { $\rcv{N_5}$ };
\node[tag] at (11.4, -2.6) { $\rcv{N_5}$ };
\node[tag] at (11.4, -3.6) { $\rcv{N_5}$ };

\draw[message,color=green!50!black] (13.75, -3) -- (14.5, -0);
\draw[message,color=green!50!black] (13.75, -3) -- (14.5, -1);
\draw[message,color=green!50!black] (13.75, -3) -- (14.5, -2);
\draw[message,color=green!50!black] (13.75, -3) -- (14.5, -4);

\node[tag] at (13.5, -3.4) { $\snd{P_3}; \snd{C_5}$ };

\end{tikzpicture}

\caption{An unhappy path scenario: A faulty replica partitions the correct replica by 
  sending two different proposals. After that, it sends a timeout quorum certificate that contains two 
  subquorums of equal size}\label{fig:unhappy2}

\end{figure}
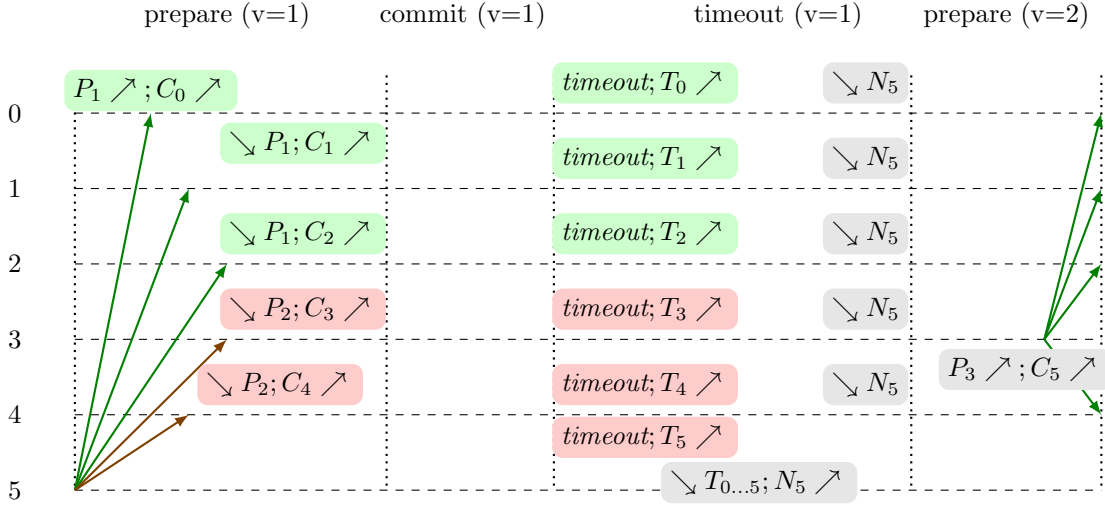 

In this case, it is crucial to determine whether any correct replicas have already committed to a previous proposal. This distinction affects how the protocol proceeds to ensure safety and avoid forks. There are two possible scenarios based on the collected \texttt{TimeoutVote} messages:

\begin{enumerate}

\item \textbf{A previous proposal may have been secretly finalized:} At least $n-3f$ replicas have voted for a previous proposal, as can be verified by comparing the replicas highest commit votes shared in their \textit{TimeoutVote} messages, and no honest replica has observed a commit quorum certificate for that proposal. This means that there might have been $n-f$ commit votes broadcast for that proposal (the observed $n-3f$ high commit votes plus $f$ honest replicas that haven't participated in the timeout quorum certificate plus $f$ malicious replicas that participated in the timeout quorum certificate and might be lying about their high commit votes), and so a replica might be hiding a commit quorum certificate for that proposal.

\item \textbf{No previous proposal has been secretly finalized:} Exactly the reverse condition of the previous scenario. Either at least $n-3f$ replicas have voted for the previous proposal but we also have the corresponding commit quorum certificate or less than $n-3f$ replicas have voted for the previous proposal. Either way we are able to prove that the previous proposal could not have been finalized.
\end{enumerate}

In the first case, the leader cannot conclusively determine whether some silent or malicious replicas might have committed the previous proposal.

To maintain safety, the protocol enforces the following rule: If the leader 
used a \texttt{TimeoutQC} as a justification for switching to the current view, 
and this \texttt{TimeoutQC} includes a subset of $n-3f$ votes (called a 
\textit{subquorum}) that favor the previous proposal, the leader must re-
propose the previous block. This strategy ensures that any potential 
commitments made by honest replicas are honored, preventing the possibility of 
conflicting commits at the same block number.

The rationale behind this rule is elaborated in the safety proofs provided in Section \ref{sec:proofs}. By requiring the leader to re-propose under these conditions, the protocol preserves the safety property, ensuring that all honest replicas eventually agree on the same sequence of committed blocks.

In the first case, the leader also uses a \texttt{TimeoutQC} as a justification for switching to the current view but is free to propose a new block. Figure~\ref{fig:unhappy1} shows this scenario. Replica~0 is a
correct leader, which proposes~$P_1$. All correct replicas receive the
proposal, but fail to receive the commit votes from the peers, due to
asynchrony. As a result, all correct replicas time out. As the correct 
replicas keep sending their latest messages
of each kind, the five correct replicas eventually receive timeout votes
and go to view~2. There, replica~3 becomes the leader, and it re-
proposes~$P_1$.

Figure~\ref{fig:unhappy2} illustrates the second case, which demonstrates
the power of Byzantine faults. The faulty replica~5 sends two conflicting
proposals~$P_1$ and $P_2$ to two subsets of the correct replicas. As a result,
three correct replicas time out and send
timeout votes supporting $P_1$, whereas two correct replicas time out and send timeout votes supporting $P_2$.
The faulty replica collects all these timeout votes, denoted with~$\rcv{T_{0\dots5}}$,
adds its own timeout vote supporting~$P_2$ and sends
a new-view message~$\snd{N_5}$ with a valid timeout quorum certificate, which has timeout votes that are
equally distributed for~$P_1$ and~$P_2$. In this case, the timeout quorum contains an evidence of having
two subquorums of $n - 3f$ votes. Hence, neither of the block hashes in those subquorums should be re-proposed.
In the next view, the new correct leader~3 discards the votes for~$P_1$ and~$P_2$ and proposes a completely new block, as there was no prevalent majority for either of the two proposals (if there are two subquorums in a timeout certificate, then it is mathematically impossible for either of the two proposals to have been committed).

\subsection{Expected Protocol Properties}\label{sec:properties}

As expected from a consensus algorithm, ChonkyBFT satisfies the following three properties:

\begin{itemize}
    \item \textbf{Safety.} No two correct replicas commit different blocks for the same block number.

    \item \textbf{Validity.} If a correct replica commits a block, this block must be proposed earlier and
      pass block verification.

    \item \textbf{Liveness.} For every block number~$n \in \mathbb{N}_0$, if a correct replica
      can propose a block for~$n$, under the assumption of partial synchrony,
      all correct replicas eventually commit a block for~$n$.
\end{itemize}

\section{Basic ChonkyBFT Specification}

This section specifies the basic consensus protocol. By "basic", we refer to a protocol where all replicas are treated equally, each contributing exactly one vote, irrespective of any stake-weighted considerations. We also exclude garbage collection and gossip protocol from the specification.

\subsection{System Model}
The following conditions hold:
\begin{enumerate}
    \item There is a set of $n$ replicas numbered $1, \dots, n$, of which at most $f$ are Byzantine faulty, such that $n \geq 5f+1$. All other replicas are correct.
    
    \item The network operates within the partial-synchrony model \cite{dwork1988consensus}, where communication among correct replicas may be fully asynchronous before the Global Stabilization Time (GST), unknown for all replicas, and becomes synchronous after GST.
    
    \item Network communication is implemented through reliable, authenticated, and point-to-point channels. A correct replica receives a message from another correct replica if and only if the latter explicitly sends that message to the former.
    
    \item An adversary controls communication channels within the partial-synchrony model and can corrupt up to $f$ replicas.

    \item A digital signature scheme is a tuple of three algorithms: \textit{KeyGen}, \textit{Sign}, \textit{Verify}. Each replica $i \in \{1, \dots, n\}$, possesses a public-private key pair $(priv_{i}, pub_{i})$ generated using the $KeyGen$ algorithm and is identified by its public key $pub_{i}$. Each replica knows all other replicas' public keys (identifiers). The $i$-th replica can use its private key to create a signature $sig = Sign(priv_{i}, m)$ on a message $m$.  Any replica can verify the signature $sig$ using the corresponding public key and the function $Verify(sig, m, pub_{i})$. 
    
    \item An aggregate signature scheme (e.g., BLS \cite{bls_signature}) is a tuple of four algorithms: \textit{KeyGen}, \textit{Sign}, \textit{Aggregate}, \textit{AggregateVerify}. The first two algorithms provide a cryptographic interface similar to traditional digital signatures. Each replica possesses a pair of keys $(priv_{i}, pub_{i})$ generated using the $KeyGen$ algorithm and shares its public key $pub_{i}$. The $i$-th replica can use its private key to create a signature $sig = \textit{Sign}(priv_{i}, m)$ on a message $m$. Each replica can aggregate a collection of signatures $sig_{1}, \dots, sig_{n}$ into a single signature $sig = \textit{Aggregate}(sig_{1}, \dots, sig_{n})$. Any replica can verify this aggregated signature on the messages $m_{1}, \dots, m_{n}$ using the public keys $pub_{1}, \dots, pub_{n}$ of the corresponding replicas and the function $\textit{AggregateVerify}((pub_{1}, \dots, pub_{n}), (m_{1}, \dots, m_{n}), sig)$.

    \item All cryptographic primitives used, such as digital signatures, aggregate signatures, and hash functions, are perfectly secure.
\end{enumerate}

\subsection{Notation}

With $\mathbb{N}_0$, we denote the set of numbers $0, 1, 2, \dots$.
We will also use the special value $\mathit{None}$ that can be assigned to any variable,
even if it is outside of all value domains.

The protocol operates in discrete views from the set of $\mathbb{N}_{0}$. Let \textit{view\_number} denote the \textit{view number}. 

A \textit{block} is an element of the set of all blocks. For our purposes, it suffices to
say that the set of all blocks is a subset of $\mathbb{N}_{0}$.
In practice, a block contains the relevant blockchain data, such as the list of signed transactions, the application state, the hash of the parent block, etc.

A correct replica cycles through three phases: \textit{Prepare}, \textit{Commit}, and 
\textit{Timeout}. With $\textit{Phase}$, we denote the set of these three phases.

We represent \textit{protocol messages} as tagged (named) tuples. We say that a message has a tag $T$ if it is a tuple tagged by $T$. More precisely,
when sending a message with a tag $T$ and fields $e_1, e_2, \dots, e_n$, we write
it as $\langle T, e_1, e_2, \dots, e_n \rangle$.
When receiving a message,
we write $\langle T, x_1, x_2, \dots, x_n \rangle$ to denote that the message must
be tagged with $T$ and its fields are deconstructed into variables
$x_1, x_2, \dots, x_n$ of the tag $T$. Finally, we use \underline{{ }{ }} to denote that the specific field may contain any value. 

For example, a tuple $\langle \text{CommitVote}, 2, 3, 4 \rangle$ corresponding to the message:
\[
\langle \text{CommitVote}, view, block\_number, block\_hash \rangle
\]

would indicate a commit vote for view 2, block number 3, and block hash 4. When receiving  
$\langle \text{CommitVote}, view, block\_number, block\_hash \rangle$, the placeholders
\textit{view}, \textit{block\_number}, and \textit{block\_hash} represent the values 2, 3, and 
4, respectively. 

Furthermore, receiving:
\[
\langle \text{CommitVote}, view, block\_number, \textunderscore \, \rangle
\]
indicates that the \textit{block\_hash} field is irrelevant.

Let a replica with number $i \in N$ and public key $pub_{i}$ be in the current view 
$view\_number$. The replica’s \textit{commit vote} for a block $b$, is represented as:

\[
\langle \text{CommitVote}, view\_number, block\_number, block\_hash \rangle,
\]

where $block\_number$ is a block number the replica is committing to, and $block\_hash = hash(b)$. 

A \textit{signed commit vote} of a commit vote $vote$ produced by the replica $i$ is defined as 
a tuple:

\[
\langle \text{SignedCommitVote}, vote, sig \rangle,
\]

where $sig = Sign(priv_{i}, vote)$.  That signed commit vote and its corresponding commit vote $vote$ are valid if $\textit{Verify}(vote, sig) = valid$.

We denote a \textit{commit quorum certificate} (QC) of a block $b$ with a
$\langle  \text{CommitQC}, vote, agg\_sig \rangle$, where $agg\_sig$ is an aggregate signature produced by aggregating the signatures of the same commit vote $\langle \text{CommitVote}, view, b.block\_number, b.block\_hash \rangle$.
Additionally, we assume that the identities of the signing replicas can be extracted from
a \textit{CommitQC}. We write $\textit{Signers}(qc)$ to denote the public keys
of the replicas whose signatures were aggregated into $qc.agg\_sig$. Further,
we say that a certificate $qc$ is \emph{valid}, if:

\begin{itemize}
    \item $|\textit{Signers}(qc)| \ge n - f$, and
    \item $\textit{AggregateVerify}(\textit{Signers}(qc), qc.vote, qc.agg\_sig)$
      returns true.
\end{itemize}

The commit vote with the maximum view number the replica has ever signed is called the \textit{high vote} and denoted with \textit{high\_vote}.

The commit QC with the maximum view number the replica has ever observed is called the \textit{high commit QC} and is denoted by \textit{high\_commit\_qc}. The view number of the highest commit quorum certificate that this replica has observed is called \textit{highest commit QC view number} and is denoted by \textit{high\_commit\_qc\_view}.

A tuple $\langle \text{TimeoutVote}, view, high\_vote, high\_commit\_qc\_view \rangle$ is called a timeout vote, and a signed timeout vote of a timeout vote $vote$ produced by the replica $i$ is defined as:

\[
\langle \text{SignedTimeoutVote}, vote, high\_commit\_qc, sig \rangle,
\]

where $sig = Sign(priv_{i}, vote)$ and $high\_commit\_qc$ is the commit QC with the highest view that this replica has observed. The signed timeout vote $st$ and its corresponding timeout vote $st.vote$ are valid if both of the following conditions hold:

\begin{enumerate}
    \item $\textit{Verify}(st.vote, st.sig) = valid$,
    \item $st.high\_commit\_qc$ is \texttt{None}, or is valid and $st.vote.high\_commit\_qc\_view = st.high\_commit\_qc$. 
\end{enumerate}

A \textit{timeout quorum certificate} of a view $v$ is a tuple:

\[
\langle \text{TimeoutQC}, votes, high\_commit\_qc, agg\_sig),
\]

where $votes$ is a set of timeout votes, $high\_commit\_qc$ is the commit quorum certificate with the highest view that all replicas in this QC have observed, and $agg\_sig$ is an aggregate signature produced by aggregating the signatures of the timeout votes $votes$ for this view $v$ from at least $n - f$ replicas.
Similar to the case of \texttt{CommitQC}, for a timeout quorum certification $qc$,
we denote with $\textit{Signers}(qc)$ the set of the public keys of the replicas
that have signed~$qc$.

A timeout quorum certificate $qc$ is \emph{valid}, if all of the following conditions hold
true:

\begin{enumerate}
    \item $|\textit{Signers}(qc)| \ge n - f$,
    \item $\textit{AggregateVerify}(\textit{Signers}(qc), qc.votes, qc.agg\_sig)$
      returns true,
    \item All votes have the same view:

    $\forall pub_1, pub_2 \in keys(qc.votes): qc.votes[pub_1].view = qc.votes[pub_2].view$

    \item If $qc.high\_qc$ is \texttt{None}, then all votes in $qc.votes$
      have the field $high\_commit\_qc\_view$ set to \texttt{None}.

    \item If $qc.high\_qc$ is different from \texttt{None}, then
      $qc.high\_qc.vote.view$
      is the highest view among the views $high\_commit\_qc\_view$
      of all votes in $qc.votes$. Moreover, $qc.high\_qc$ is valid.
\end{enumerate}

In the following, we use $view(qc)$ to denote the common view of the votes in
a valid timeout quorum certificate~$qc$. By slightly abusing the notation,
we use~$view(qc)$ to denote $qc.vote.view$ for a commit quorum certificate~$qc$.

A \textit{justification} is either a commit or a timeout QC. A justification is
valid if its underlying quorum certificate is valid. We use justifications and
quorum certificates interchangeably.

A \textit{committed block} is a pair $\langle b, qc \rangle$, for a block $b$
and a commit QC $qc$ that satisfy the following: $hash(b) = qc.vote.block\_hash$.

A \textit{proposal} is a tuple  $\langle \text{Proposal}, \textit{block}, 
\textit{justification}, \textit{sig} \rangle$. A proposal is called valid if its signature and 
justification are valid.

A \textit{new view} is a tuple $\langle \text{NewView}, \textit{justification}, sig \rangle$ 
denoted by $new\_view$, called valid if $new\_view.sig$ and $new\_view.\textit{justification}$
are valid.

A \textit{replica state} is a tuple that contains the following fields:

\begin{itemize}
    \item \texttt{view} is the replica's view, initially \texttt{None},
    \item \texttt{phase} is the replica's phase, initially \texttt{Prepare},
    \item \texttt{high\_vote} is the highest vote sent, initially \texttt{None},
    \item \texttt{high\_commit\_qc} is the highest observed commit QC, initially \texttt{None},
    \item \texttt{high\_timeout\_qc} is the highest observed timeout QC, initially \texttt{None},
    \item \texttt{committed\_blocks} is the list of committed blocks, initially the empty list \texttt{[]}.
\end{itemize}

We often refer to the state components of a replica $r$ using the dot notation. For example, \texttt{r.view} refers to the view of replica $r$ in the current state.

\subsection{Protocol Specification in Pseudocode}

We introduce the protocol as pseudocode in Listings~\ref{lst:defs}--\ref{lst:replica2}.
Listings~\ref{lst:defs}--\ref{lst:defs2} contain several key definitions related to quorum certificates.
These definitions do not depend on the protocol state.
Listing~\ref{lst:proposer} captures the proposer logic for a replica~$r$.
The replica acts as the proposer in a view~$v$ only when~$r$ is the leader of the view~$v$, or
$r = \textit{leader}(v)$. Listings~\ref{lst:replica}--\ref{lst:replica2} capture the logic of
a replica~$r$ in every view, independently of whether the replica acts as the leader in that view or not.

\subsubsection{Pseudocode Conventions}

We present the state transitions by replicas (acting as proposers and non-proposers)
in terms of event handlers of the following shape:

\begin{lstlisting}
upon <event> (when <condition>)*:
  <body>
\end{lstlisting}

An event is one of the five kinds:

\begin{enumerate}
  \item The event \texttt{start} that is triggered once when the proposer or replica code starts.
  
  \item The event \texttt{timeout} that is triggered once the associated timer expires.

  \item The event \texttt{*} that is triggered as soon as the handler's condition holds true.

  \item The event \texttt{<Msg, ...>} that is triggered as soon as the replica receives a
    message~\texttt{Msg} of the specified shape~\texttt{<Msg, ...>}. Note that every message is
    processed at most once.

  \item The event \texttt{\{ <Msg, ...> \} as $ms$} that is triggered as soon as the replica receives a
    set of messages~$ms$ that all have the specified shape~\texttt{<Msg, ...>} and satisfy
    the handler's condition. In this case, every message in the set~$ms$ is considered to be processed.
\end{enumerate}

The event conditions are written over the current replica state and, optionally, over the fields
of the received messages. The event handler's body contains pseudocode that further evaluates
the current replica state and computes its next state. The next state is computed from the current
state in terms of parallel assignments ${\ell_1, \dots, \ell_n \gets e_1, \dots, e_n}$, meaning that
every expression~$\ell_i$, e.g., a variable name or a record field, has the value~$e_i$ after the
assignment, for $1 \le i \le n$.

Replicas can only read and modify their own state (as proposers and non-proposers).
Replicas organize their distributed computation by broadcasting and receiving messages.

Every event handler is processed \textit{indivisibly}. Thus, no concurrency issues arise
from handlers' execution on the same machine. Besides the event conditions, written as
\texttt{when <condition>}, we introduce two additional statements:

\begin{lstlisting}
upon <event> (when <condition1>)*:
  ...
  assume(condition2)
  assert(condition3)
  ...
\end{lstlisting}

The meaning of \texttt{assume(condition2)} is that the respective event handler is only processed
when \texttt{condition2} holds true during the handler's execution. Otherwise, the event is discarded.
Hence, the semantics of \texttt{assume} is exactly the same as the semantics of \texttt{when}.
Technically, we could replace all \texttt{when}-conditions with \texttt{assume}, but we keep both for
presentation purposes. On the contrary, \texttt{assert(condition3)} specifies a protocol invariant,
which under no circumstances should be violated. Hence, assertions should be understood as additional
explanations about the expected code behavior.

\subsubsection{Pseudocode Explained}

\paragraph{Sending proposals.}
The computation of the consensus algorithm is organized in views. The correct view leader~\texttt{leader(v)}
sends their proposal in lines~\ref{line:propose-begin}--\ref{line:propose-end}. A correct proposer has to re-propose a previously proposed block,
if the following conditions hold true:

\begin{itemize}
  \item The leader for the current view has to justify his proposal with a
    timeout QC~$\tqc$ (i.e. the previous view timed out).

  \item The timeout QC~$\tqc$ contains a high vote, that is,
   it contains exactly one subquorum of replica votes for a given block.

  \item The timeout QC~$\tqc$ does not contain a high commit QC that
    corresponds to the high vote, that is, the high vote and the high commit 
    QC are issued for different block numbers.
\end{itemize}

In all other cases, the proposer is free to propose a new block. This logic is
captured in the definition of~\texttt{get\_implied\_block}, see
lines~\ref{line:implied-begin}--\ref{line:implied-end}.

\paragraph{Receiving proposals.} A correct replica receives a proposal in
lines~\ref{line:on-proposal-begin}--\ref{line:on-proposal-end}. Similar to the view proposer,
a replica uses the definition~\texttt{get\_implied\_block} to send its commit vote.
Importantly, the block contents are only received along with the proposal if a new block
is proposed. The rest of the computation is done over block hashes. Every correct replica processes
at most one proposal per view.

\paragraph{Receiving commit votes.} In lines~\ref{line:on-commit-begin}--\ref{line:on-commit-end},
a correct replica handles the event of receiving commit votes for the same view, block number, and block
hash from at least~$n - f$ replicas (correct or faulty). If this is the case, the replica constructs
a commit quorum certificate by aggregating the signatures over the votes. The replica sends
this certificate by broadcasting a \texttt{NewView} message and advances to the next view.

In this case, the replica finalizes a block for the current block number. It adds
the block to the list of committed blocks in lines~\ref{line:process-commit-start}--\ref{line:process-commit-end}.
It is also possible that at most~$f$ correct replicas catch up with a faster quorum by receiving~$n - f$
commit votes from the quorum participants without having received the contents of the finalized block
in a proposal. Even in this case, the late replicas store the block hash as finalized. However, they have to retrieve
the block contents with an additional synchronization mechanism called \textit{block fetcher}.

\paragraph{Sending and receiving timeout votes.} If a correct replica stays in the current view for the duration
of ~$\timeout$ time units, it broadcasts a timeout vote in
lines~\ref{line:on-send-timeout-start}--\ref{line:on-send-timeout-end}. By doing so, the replica messages
other replicas that no timely progress has been made and the replicas have to switch to the next view,
where another proposer will send their proposal.

When a replica receives timeout votes for the same view from at least~$n - f$ 
replicas (correct or faulty),
it constructs a timeout quorum certificate by aggregating the signatures over 
the votes. Importantly, the replicas find a commit quorum certificate for the 
highest view among the  quorum certificates sent
along with the timeout votes. This is essential for the proposer of the next
view to avoid re-proposing a finalized block. The replica sends
the timeout quorum certificate by broadcasting a \texttt{NewView} message and 
advances to the next view.

\paragraph{Receiving new view messages.} Finally, a replica may receive a quorum certificate, either
a commit QC or a timeout QC, from another replica. By doing so, the replica avoids receiving~$n - f$
commit or timeout votes and constructing its own quorum certificates. Additionally, if it receives a commit
QC for the current block number, it may immediately finalize a block. In both cases, the replica, if delayed,
will immediately jump to the next view of the received quorum certificate.

\paragraph{Bootstrapping logic.} Replicas start in view 0, and thus, they do not have a quorum certificate
for the previous view. To work around this problem, the replicas immediately send timeout votes upon starting in
lines~\ref{line:on-send-timeout-start}--\ref{line:on-send-timeout-end}. As a result, under synchrony
in view~0, correct replicas are expected to produce a timeout QC and proceed to view~1. It is possible to use other bootstrapping methods, if desired.

\begin{lstlisting}[float, columns=fullflexible, caption={Key definitions over quorum certificates}, label=lst:defs]
// Return the highest commit quorum certificate
def max(qc1, qc2):
  if qc1 is None:
    return qc2
  else if qc2 is None:
    return qc1
  else if qc1.vote.view >= qc2.vote.view:
    return qc1
  else:
    return qc2

// the view of a commit or timeout QC
def view(qc):
  if qc is CommitQC:
    return qc.vote.view
  if qc is TimeoutQC:
    // get the common view of the timeout votes (a valid QC has only one such view)
    return qc.votes[pk].view for some pk in keys(qc.votes)

// Get a high vote for a subquorum of >= n - 3f timeout votes in a TimeoutQC,
// if there is exactly one such subquorum. Otherwise, return None.
def high_vote(qc):
  // collect all available high votes in qc.votes
  let HV = (*@$ \{ pk \in keys(qc.votes) : qc.votes[pk].high\_vote \} \setminus\ \text{None} $@*)
  if (*@$\exists v_1 \in HV: |\{ pk \in keys(qc.votes): qc.votes[pk].high\_vote = v_1 \}| \ge n - 3f$@*): (*@\label{line:q1}@*)
    // there is subquorum of timeout votes that contain (*@$v_1$@*) as high vote
    if (*@$\exists v_2 \in HV \setminus \{ v_1 \}: |\{ pk \in keys(qc.votes): qc.votes[pk].high\_vote = v_2 \}| \ge n - 3f$@*):(*@\label{line:q2}@*)
      // ...and a subquorum of timeout votes that contain (*@$v_2 \ne v_1$@*) as high vote
      return None
    else:
      return v1
  else:
    return None
\end{lstlisting}

\begin{lstlisting}[float, columns=fullflexible, firstnumber=last,
  caption={Key definitions over quorum certificates (continued)}, label=lst:defs2]
// Correct replicas coordinate on how they either advance the block number,
// or re-propose the block that they have collectively voted upon in the past.
// This is captured with a quorum certificate for the view change (justification):
// Both the view leader and the non-leaders use this function to decide,
// whether to proceed with the next block number, or re-propose the hash
// of the previously voted upon block.
// Returns the tuple <block_number, block_hash>.
def get_implied_block(qc): (*@\label{line:implied-begin}@*)
  if qc is CommitQC:
    // The previous proposal was finalized, so we can propose a new block.
    return (qc.vote.block_number + 1, None) (*@\label{line:implied-next-block}@*)
  if qc is TimeoutQC:
    let hv = high_vote(qc)
    if hv is not None and (qc.high_commit_qc is None
        or hv.block_number > qc.high_commit_qc.vote.block_number):
      // There was a proposal in the last view that might have been finalized.
      // We need to repropose it.
      return (hv.block_number, hv.block_hash) (*@\label{line:implied-reproposal}@*)
    else:
      // Either the previous proposal was finalized or we know for certain
      // that it couldn't have been finalized. Either way, we can propose
      // a new block.
      if qc.high_commit_qc is None:
        return (0, None)
      else:
        return (qc.high_commit_qc.vote.block_number + 1, None)  (*@\label{line:implied-end}\label{line:implied-next-block2}@*)
\end{lstlisting}

\begin{lstlisting}[float, columns=fullflexible, caption={Proposer logic}, firstnumber=last, label=lst:proposer]
var cur_view (*@$\gets$@*) 0 // an internal variable to keep track of the view changes

upon start: // of proposer for replica (*@$r$@*)
  cur_view (*@$\gets$@*) r.view
  
upon * when r.view > cur_view and (*@$r$@*) is leader(r.view): (*@\label{line:propose-begin}@*)
  cur_view (*@$\gets$@*) r.view
  qc (*@$\gets$@*) create_justification_qc()
  (block_number, block_hash) (*@$\gets$@*) get_implied_block(qc)
  assume(replica.committed_blocks has a block for all n (*@$\leq$@*) block_number) (*@\label{line:proposer-has-blocks}@*)
  let b = None if block_hash is None, else create_proposal(block_number)
  broadcast <Proposal, b, qc> (*@\label{line:propose-end}@*)

def create_justification_qc():
  assert(r.high_commit_qc is not None or r.high_timeout_qc is not None)
  if r.high_commit_qc is not None and (r.high_timeout_qc is None
      or view(r.high_commit_qc) (*@$\geq$@*) view(r.high_timeout_qc)):
    return r.high_commit_qc
  else:
    return r.high_timeout_qc
\end{lstlisting}

\begin{lstlisting}[float,columns=fullflexible, caption={Replica logic}, firstnumber=last, label=lst:replica]
var cur_view (*@$\gets$@*) 0 // an internal variable to keep track of the view changes
var timer // an internal timer

upon start: // of replica (*@$r$@*)
  timer.start((*@$\timeout$@*))
  cur_view (*@$\gets$@*) r.view

upon <Proposal, block, qc>(*@\label{line:on_proposal}\label{line:on-proposal-begin}@*)
    when (r.phase = Prepare and r.view = view(qc) + 1) or view(qc) + 1 > r.view
    when signer is leader(view(qc) + 1):
  let (block_number, block_hash) = get_implied_block(qc)
  assume(replica.committed_blocks has a block for all n (*@$\leq$@*) block_number)(*@\label{line:on-proposal-has-blocks}@*)
  // Check if this is a reproposal or not, and do the necessary checks
  assume((block_hash is None) = (block is not None))
  if block_hash is None:
    assume(r.verify_block(block_number, block)) (*@\label{line:on-proposal-verify-block}@*)
    cache block for (block_number, hash(block))
  let bh = hash(block) if block_hash is None, else block_hash
  let vote = <CommitVote, view(qc) + 1, block_number, bh> (*@\label{line:create-commit}@*)
  r.phase, r.view, r.high_vote (*@$\gets$@*) Commit, view(qc) + 1, vote (*@\label{line:set-high-vote}@*)
  if qc is CommitQC:
    r.process_commit_qc(qc)
  if qc is TimeoutQC:
    r.process_commit_qc(qc.high_commit_qc)
    r.high_timeout_qc (*@$\gets$@*) max(qc, r.high_timeout_qc)
  broadcast <SignedCommitVote, vote> (*@\label{line:send-commit}\label{line:on-proposal-end}@*)

upon { <SignedCommitVote, <CommitVote, view, block_number, block_hash>, _> } as (*@$vs$@*) (*@\label{line:on-commit-begin}@*)
     when (*@$|vs| \ge n - f$@*) and view (*@$\geq$@*) r.view:
  let agg_sig = Aggregate(vs) // signatures from (*@$vs$@*)
  let vote = <CommitVote, view, block_number, block_hash>
  let qc = <CommitQC, vote, agg_sig> (*@\label{line:commit-qc}@*)
  r.process_commit_qc(qc)
  start_new_view(view + 1) (*@\label{line:on-commit-end}@*)

upon (start or timeout) when r.phase (*@$\ne$@*) Timeout: (*@\label{line:on-send-timeout-start}@*)
  // Update our state so that we can no longer vote commit in this view
  r.phase (*@$\gets$@*) Timeout
  broadcast <TimeoutVote, r.view, r.high_vote, r.high_commit_qc> (*@\label{line:send-timeout}\label{line:on-send-timeout-end}@*)

upon { <SignedTimeoutVote, <TimeoutVote, view, _, _>, _> } as (*@$vs$@*) (*@\label{line:on-timeout-start}@*)
     when (*@$|vs| \ge n - f$@*) and view (*@$\geq$@*) r.view:
  let agg_sig = Aggregate(vs) // signatures from (*@$vs$@*)
  let votes = Map[ public_key(v) (*@$\rightarrow$@*) v.vote for v in (*@$vs$@*) ]
  // choose one of the commit QCs with the highest view
  let high_commit_qc =
    v.high_commit_qc for some v in (*@$vs$@*):
      for w in vs: v.high_commit_qc.vote.view (*@$\geq$@*) w.high_commit_qc.vote.view
  let qc = <TimeoutQC, votes, agg_sig, high_commit_qc>
  r.process_commit_qc(high_commit_qc)
  r.high_timeout_qc (*@$\gets$@*) max(qc, r.high_timeout_qc)
  start_new_view(view + 1) (*@\label{line:on-timeout-end}@*)
\end{lstlisting}
    
\begin{lstlisting}[float,columns=fullflexible, caption={Replica logic (continued)}, firstnumber=last, label=lst:replica2]
upon <NewView, qc> when view(qc) + 1 (*@$\geq$@*) r.view (*@\label{line:on-new-view-start}@*)
  if qc is CommitQC:
    r.process_commit_qc(qc)
  if qc is TimeoutQC:
    r.process_commit_qc(qc.high_commit_qc)
    r.high_timeout_qc (*@$\gets$@*) max(qc, r.high_timeout_qc)
  if view(qc) + 1 > r.view:
    start_new_view(view(qc) + 1) (*@\label{line:on-new-view-end}@*)

upon * when r.view > cur_view:
  // If the view has increased before the timeout, we reset the timer
  cur_view (*@$\gets$@*) r.view
  timer.reset() (*@\label{line:reset-timer}@*)

def start_new_view(view): (*@\label{line:start-new-view-start}@*)
  r.phase, r.view (*@$\gets$@*) Prepare, view
  broadcast <NewView, create_justification_qc()> (*@\label{line:send-new-view}@*)(*@\label{line:start-new-view-end}@*)

def process_commit_qc(qc): (*@\label{line:process-commit-start}@*)
  if qc is not None:
    r.high_commit_qc (*@$\gets$@*) max(qc, r.high_commit_qc)
    if r has a cached block (*@$b$@*) for (qc.vote.block_number, qc.vote.block_hash):
      if qc.vote.block_number = r.committed_blocks.length:
        r.committed_blocks (*@$\gets$@*) r.committed_blocks.append((*@$b$@*))
        commit((*@$b$@*)) (*@\label{line:commit}\label{line:process-commit-end}@*)
\end{lstlisting}

\section{Correctness Proofs} \label{sec:proofs}

\subsection{Proofs of safety}

Before showing the safety of the protocol, we prove several protocol invariants on quorum certificates.
We start with a lemma about set cardinalities, which is needed to reason about quorums and subquorums:

\begin{lemma}\label{lem:quorum-intersections}
    Given two sets~$A$ and~$B$, from the inequalities $|A| \ge \alpha$,
    $|B| \ge \beta$, $|A \cup B| \le \gamma$,
    it follows that $|A \cap B| \ge \alpha + \beta - \gamma$.
\end{lemma}

\begin{proof}
  Let $C$ be $A \cap B$. Rewrite the lemma assumptions as follows:

  \begin{align}
      |(A \setminus C) \cup C| &\ge \alpha\label{eq:AC}
      \\
      |(B \setminus C) \cup C| &\ge \beta\label{eq:BC}
      \\
      |(A \setminus C) \cup (B \setminus C) \cup C| &\le \gamma\label{eq:ABC}
  \end{align}

  By summing up the inequalities~(\ref{eq:AC}) and~(\ref{eq:BC})
  and subtracting the inequality~(\ref{eq:ABC}), we get:

  \begin{equation}
    |(A \setminus C) \cup C| + |(B \setminus C) \cup C| - |(A \setminus C) \cup (B \setminus C) \cup C| \ge \alpha + \beta - \gamma \label{eq:ACBCABC}
  \end{equation}

  Now, notice that $C$, $A \setminus C$, and $B \setminus C$ are pairwise 
  disjoint, since~$C = A \cap B$.
  Hence, we have $|(A \setminus C) \cup C| = |A \setminus C| + |C|$,
  $|(B \setminus C) \cup C| = |B \setminus C| + |C|$, and $|(A \setminus C) 
  \cup (B \setminus C) \cup C| = |A \setminus C| + |B \setminus C| + |C|$.
  By applying those equalities to inequality~(\ref{eq:ACBCABC}), we arrive at:

  \begin{equation}
      |A \setminus C| + |C| + |B \setminus C| + |C| - |A \setminus C| - |B \setminus C| - |C| \ge \alpha + \beta - \gamma\label{eq:C}
  \end{equation}

  By eliminating equal terms in~(\ref{eq:C}), we obtain $|C| \ge \alpha + \beta - \gamma$.
  Thus, $|A \cap B| \ge \alpha + \beta - \gamma$.
\end{proof}

In particular, Lemma~\ref{lem:quorum-intersections} gives us the following corollary for sets of replica
identities:

\begin{corollary}\label{cor:correct}
  For every set $X \subseteq \Corr\,\cup \Faulty$, if $|X| \ge \alpha$, then
  $|X \cap \Corr\,| \ge \alpha - f$.
\end{corollary}

\begin{proof}
    Since $|\Corr| = n - f$ and $|\Corr\,\cup \Faulty\,| = n$, we have $|\Corr\,| \ge n - f$ and $|X \cup \Corr\,| \le n$. Apply Lemma~\ref{lem:quorum-intersections} with~$A = X$, $B = \Corr$, $\alpha = \alpha$, $\beta = n - f$, $\gamma = n$.
\end{proof}

Further, we apply Lemma~\ref{lem:quorum-intersections} and Corollary~\ref{cor:correct}, to reason
about intersections of quorum signers:

\begin{corollary}\label{cor:commit-intersections}
  Let $\cqc_1$, $\cqc_2$ be two valid commit quorum certificates. Then,
  $|\,\signers{\cqc_1}\, \cap\, \signers{\cqc_2}\,\cap\, \Corr{}\,| \ge n - 3f$.
\end{corollary}

\begin{proof}
  Since $\cqc_1$ and $\cqc_2$ are valid commit quorum certificates, each of them must be
  signed by at least $n - f$ replicas, that is, $|\signers{\cqc_1}| \ge n - f$ and
  $|\,\signers{\cqc_2}\,| \ge n - f$. Further, $|\signers{\cqc_1}\,\cup\,\signers{\cqc_2}| \leq n$,
  as $n$ is the total number of replicas. We apply Lemma~\ref{lem:quorum-intersections}
  with~$A = \signers{\cqc_1}$, $B = \signers{\cqc_2}$, $\alpha = \beta = n - f$, $\gamma = n$, and 
  immediately obtain that:
\begin{equation}
    |\,\signers{\cqc_1}\, \cap\, \signers{\cqc_2}\,| \ge n - 2f \label{eq:cqc1-2}
  \end{equation}

  By applying Corollary~\ref{cor:correct} to~(\ref{eq:cqc1-2}),
  we get $|\,(\signers{\cqc_1}\, \cap\, \signers{\cqc_2})\,\cap\, \Corr{}\,| \ge n - 3f$.
\end{proof}

By replacing the commit quorum certificate~$\cqc_2$ with a timeout quorum certificate $\tqc_2$ in
Corollary~\ref{cor:commit-intersections} and its proof, we immediately obtain a similar result
for the intersections of commit QCs and timeout QCs:

\begin{corollary}\label{cor:commit-timeout-intersections}
  Let $\cqc_1$ and $\tqc_2$ be a valid commit quorum certificate and a valid timeout
  quorum certificate, respectively. Then,
  $|\,\signers{\cqc_1}\, \cap\, \signers{\tqc_2}\,\cap\, \Corr{}\,| \ge n - 3f$.
\end{corollary}

Now we are in a position to prove the first result about the commit quorum certificates
that are produced in the same view:

\begin{lemma}\label{lem:cqc-one-vote}
  Consider an arbitrary execution of the protocol. Let $\cqc_1$, $\cqc_2$ be two valid commit
  quorum certificates that are issued by replicas (correct or faulty). The following holds true
  for $\cqc_1$, $\cqc_2$: If $\cqc_1.vote.view = \cqc_2.vote.view$, then $\cqc_1.vote = \cqc_2.vote$.
\end{lemma}

\begin{proof}
  We apply Corollary~\ref{cor:commit-intersections} to~$\cqc$ and~$\cqc'$,
  to obtain $|\,(\signers{\cqc_1}\, \cap\, \signers{\cqc_2})\,\cap\, \Corr{}\,| \ge n - 3f$.
  By the protocol assumption, $n > 5f$, and thus $n - 3f > 0$. We conclude that:
  
  \begin{equation}
    (\signers{\cqc_1} \cap \signers{\cqc_2}) \cap \Corr{} \ne \emptyset
        \label{eq:cqc-intersection-non-empty}
  \end{equation}

  From~(\ref{eq:cqc-intersection-non-empty}), we know that a correct replica~$r$ has
  signed $\cqc_1.vote$ and $\cqc_2.vote$. A commit vote is sent in line~\ref{line:send-commit} of
  a proposal message handler. Since this handler is only activated when $r.phase = \textit{Prepare}$
  and it changes $r.phase$ to $\textit{Commit}$, and the views of  $\cqc_1.vote$ and $\cqc_2.vote$
  coincide (the lemma assumption), we immediately arrive at the conclusion that~$r$
  sent only one commit vote in view $\cqc_1.vote.view = \cqc_2.vote.view$.
  Hence, $\cqc_1.vote$ and $\cqc_2.vote$ are identical.
\end{proof}

For the following lemma, we need an additional technical definition. For two commit votes~$cv_1$
and~$cv_2$, we say that these commit votes are \textit{equal-modulo-view}, or, $cv_1 \approx cv_2$,
if the following conditions hold true:
\begin{align*}
    cv_1.block\_number &= cv_2.block\_number \\
    cv_1.block\_hash &= cv_2.block\_hash
\end{align*}

Lemma~\ref{lem:cqc-one-vote} deals with the case of two commit quorums for the same view.
In general, we have to deal with the case of having a commit quorum in one view and timeout quorums
in later views. Essentially, we show that if there was a commit quorum on a block,
then all correct replicas should re-propose this block, as captured in line~\ref{line:implied-reproposal}
of~\texttt{get\_implied\_block}. Lemma~\ref{lem:timeout-chain} below captures this property.

First, we formulate a family of invariants over protocol executions and views.
We define system
invariants~$\textsc{Inv}(x_1, \dots, x_m)$ as statements over replica states, quorum certificates,
views, etc. Importantly, our invariants are parameterized by a view~$v \ge 0$.
We need this parameter to inductively prove the invariants, by expanding our view horizon.

\begin{definition}[Invariant $\invOne{v, s, \cqc}$]\label{def:inv1}
    Given a view~$v \ge 0$, a valid commit quorum certificate~$\cqc$,
    and a state~$s$ of a correct replica~$r \in \signers{\cqc}$ such that $s.view \le v$
    and $\cqc.vote.view \le v$, we define the invariant $\invOne{v, s, \cqc}$ as:
\begin{align}
    &s.view < \cqc.vote.view, \text{ or} \label{eq:view-under} \\
    &s.view = \cqc.vote.view \text{ and } s.phase = \text{Prepare}, \text{ or} \label{eq:view-eq} \\
    &s.high\_vote \eqmod \cqc.vote, \text{ or} \label{eq:hv-eq} \\
    &s.high\_vote.block\_number > \cqc.vote.block\_number \label{eq:block-no-above}
  \end{align}   
\end{definition}

\begin{definition}[Invariant $\invTwo{v, \cqc, \tqc}$]\label{def:inv2}
  Given a view~$v \ge 0$, a valid commit quorum certificate~$\cqc$, and a valid timeout quorum
  certificate~$\tqc$ such that $view(\tqc) \le v$. Further, consider the following conditions:
\begin{align}
    \tqc.high\_commit\_qc.vote.block\_number =&\ \cqc.vote.block\_number - 1 \label{eq:next_block} \\
    view(\tqc) \geq&\ \cqc.vote.view \label{eq:cqc_view_below}\\
    high\_vote(\tqc) &\eqmod \cqc.vote \label{eq:hv-eqmod}
  \end{align}

  We define the invariant $\invTwo{v, \cqc, \tqc}$ as follows:
  If the conditions~(\ref{eq:next_block}) and~(\ref{eq:cqc_view_below}) hold true,
  then the condition~(\ref{eq:hv-eqmod}) holds true.
\end{definition}

By assuming that the invariants~$\invOneName$ and~$\invTwoName$ hold true for all views
below a view~$v \ge 0$, we show that these invariants
hold true for the view~$v$, as well. This is done in Lemmas~\ref{lem:inv1}--\ref{lem:timeout-chain}.
These lemmas constitute the inductive step of proving that the invariants hold true
for all views~$v > 0$. The base of the induction is trivial: The protocol starts in view~0,
where it immediately produces timeout votes, so no commit quorum certificates are produced
in view~0.

\begin{lemma}\label{lem:inv1}
  Consider an arbitrary execution of the protocol and a view~$v \ge 0$. Assume that:
  
  \begin{itemize}
    \item The invariant $\invOne{u, s_u, \cqc_u}$ holds for every
      view~$u < v$, commit quorum certificate~$\cqc_u$,
      and every state~$s_u$ of every correct replica~$r_u \in \signers{\cqc}$ such that
      $s_u.view \le u$ and $\cqc_u.vote.view \le u$ in the execution.
  
    \item The invariant~$\invTwo{u, \cqc_u, \tqc_u}$ holds for every
      view~$u < v$, every commit quorum certificate~$\cqc_u$, and timeout quorum
      certificate~$\tqc_u$ such that $view(\tqc_u) \le u$.
  \end{itemize}
  
  Then, $\invOne{v, s, \cqc}$ holds for the view~$v$,
  every commit quorum certificate~$\cqc$,
  and every state~$s$ of every correct replica~$r \in \signers{\cqc}$
  such that~$s.view \le v$ and $\cqc.vote.view \le v$.
\end{lemma}

\begin{proof}
  Fix a commit quorum certificate~$\cqc$ such that $\cqc.vote.view \le v$,
  a correct replica~$r \in \signers{\cqc}$, and a local state~$s$
  of the replica~$r$ such that $s.view \le v$. We do the proof by case
  distinction on~$s.view$ and~$\cqc.vote.view$. The order of cases is important: When we consider
  a case~$i$, we assume that the above cases $1, \dots, i - 1$ do not apply.
  In each case, we have to show that one of the
  conditions~(\ref{eq:view-under})--(\ref{eq:block-no-above}) is satisfied, and thus, the
  invariant~$\invOne{v, s, \cqc}$ holds true.

  \underline{Case} $s.view < \cqc.vote.view$. The condition~(\ref{eq:view-under}) immediately holds true.

  \underline{Case} $s.view = \cqc.vote.view$ and $s.phase = \text{Prepare}$.
  The condition~(\ref{eq:view-eq}) holds true.

  \underline{Case} $s.view = \cqc.vote.view$ \text{ and } $s.phase \ne \text{Prepare}$.
  Since~$r \in \signers{\cqc}$, replica~$r$ has sent a commit vote~$\cqc.vote$ in
  line~\ref{line:send-commit}. When doing so, replica~$r$ sets its~$high\_vote$ to~$\cqc.vote$
  in line~\ref{line:set-high-vote}.
  There is no other way for~$r$ to change~$high\_vote$ in the same view. Hence,
  condition~(\ref{eq:hv-eq}) holds true, that is, $s.high\_vote \eqmod \cqc.vote$.

  \underline{Case} $s.view > \cqc.vote.view$ and $s.phase = \text{Prepare}$.
  In this case, replica~$r$ updated its~$high\_vote$ in line~\ref{line:set-high-vote} in the
  view~$s.view - 1$ or lower. Let~$s'$ be the state of~$r$ right after that last update
  of~$high\_vote$. Since $s.view - 1 < v$, we apply~$\invOne{v - 1, s', \cqc}$,
  which holds true by the lemma assumption,
  as~$v - 1 < v$, $s'.view \le s.view - 1 \le v - 1$, and~$\cqc.vote.view < s.view \le v$.

  \underline{Case} $s.view > \cqc.vote.view$ and $s.phase \ne \text{Prepare}$.
  In this case, replica~$r$ sent its commit vote in view~$s.view$ when processing
  a proposal in lines~\ref{line:on-proposal-begin}--\ref{line:on-proposal-end}.
  There are two kinds of proposals to consider:

    \begin{enumerate}
      \item The proposal carries a timeout quorum certificate~$\tqc$ for the view~$s.view - 1$
        as a proposal justification. Hence, we immediately apply~$\invTwo{s.view - 1, \cqc, \tqc}$,
        which holds true by the lemma assumption, since~$s.view - 1 \le v - 1$ and~$view(\tqc) = s.view - 1 \le v - 1$. We immediately conclude that inequality~(\ref{eq:cqc_view_below}) holds true.
        As $\invTwo{s.view - 1, \cqc, \tqc}$ is an implication, we have that either
        equality~(\ref{eq:next_block}) does not hold true, or constraint (\ref{eq:hv-eqmod}) holds true.
        Hence, we consider several cases on $\tqc.high\_commit\_qc.vote.block\_number$:

        \begin{enumerate}
            \item  \underline{Case}~(a):
                $tqc.high\_commit\_qc.vote.block\_number = \cqc.vote.block\_number - 1$, that is,
                equality~(\ref{eq:next_block}) holds true. Hence, by $\invTwo{s.view - 1, \cqc, \tqc}$,
                we conclude that constraint (\ref{eq:hv-eqmod}) holds true, that is,
                $s.high\_vote \eqmod \cqc.vote$, which gives us condition~(\ref{eq:hv-eq}).
                Hence, the invariant $\invOne{v, s, \cqc}$ is satisfied.
                
            \item  \underline{Case}~(b):
                $tqc.high\_commit\_qc.vote.block\_number > \cqc.vote.block\_number - 1$, that is,
                equality~(\ref{eq:next_block}) is violated\footnote{
                  While this case may seem absurd, a faulty replica can construct such a timeout quorum
                  certificate by collecting the commit votes and timeout votes of correct replicas.
                  We have reproduced this case with the model checker.
                }.
                Hence, the next block number is increased in line~\ref{line:implied-next-block2},
                and this block number is used in line~\ref{line:set-high-vote}.
                After that, we have $s.high\_vote.block\_number > \cqc.vote.block\_number$.
                Hence, inequality (\ref{eq:block-no-above}) holds true, and, thus,
                the invariant $\invOne{v, s, \cqc}$ is satisfied.
                
            \item  \underline{Case}~(c):
                $tqc.high\_commit\_qc.vote.block\_number < \cqc.vote.block\_number - 1$, that is,
                equality~(\ref{eq:next_block}) is violated. By
                Corollary~\ref{cor:commit-timeout-intersections}, there is at least
                one replica in $\signers{\cqc} \cap \signers{\tqc} \cap \Corr$. Since this
                replica has sent its commit vote that is included in~$\cqc$, it has its $high\_commit\_qc.vote.block\_number \ge \cqc.vote.block\_number - 1$; Otherwise,
                it would not be able to vote for~$\cqc$. This $high\_commit\_qc$ is sent
                by the replica as part of its timeout vote that is included in~$\tqc$.
                As a result, we have
                $tqc.high\_commit\_qc.vote.block\_number \ge \cqc.vote.block\_number - 1$.
                We arrive at a contradiction with the assumption of this case, so this case
                is impossible.
                
        \end{enumerate}

      \item The proposal carries a commit quorum certificate~$\cqc'$ for the view~$s.view - 1$
        as a proposal justification. Again, there are three cases to consider:

        \begin{enumerate}
            \item \underline{Case}~(a): $\cqc'.vote.block\_number > \cqc.vote.block\_number$.
            In this case $s.high\_vote = \cqc'.vote$, and condition~(\ref{eq:block-no-above})
            immediately holds true.

            \item \underline{Case}~(b): $\cqc'.vote.block\_number < \cqc.vote.block\_number$.
              By applying Corollary~\ref{cor:commit-intersections} to~$\cqc$ and~$\cqc'$,
              we obtain $|\signers{\cqc} \cap \signers{\cqc'} \cap \Corr| \ge n - 3f$.
              By the protocol assumptions, $n \ge 5f + 1$. Hence, the set
              $\signers{\cqc} \cap \signers{\cqc'} \cap \Corr$
              is non-empty. Let $r'$ be a correct replica in the intersection. Replica~$r'$
              has voted on~$\cqc$ in view~$\cqc.vote.view \le s.view - 1$ and on~$\cqc'$ in
              view~$r.view - 1$. This and the assumption of
              $\cqc.vote.block\_number > \cqc'.vote.block\_number$
              contradict the logic of~\texttt{get\_implied\_block} in
              lines~(\ref{line:implied-begin})--(\ref{line:implied-end}), which is never
              decreasing block numbers.

            \item \underline{Case}~(c): $\cqc'.vote.block\_number = \cqc.vote.block\_number$.
              Similar to Case~(b), we conclude that there is at least one correct replica~$r'$
              in the intersection of $\signers{\cqc} \cap \signers{\cqc'} \cap \Corr$.
              Consider the state~$s'$ of replica~$r'$ right after~$r'$ sent its commit vote
              for~$\cqc'$ and updated its~$high\_vote$ in view~$s.view - 1$. Since $s.view - 1 < v$,
              $s'.view = s.view - 1 < v$, and~$\cqc'.vote.view = s.view - 1 < v$,
              we apply the lemma assumption to conclude that~$\invOne{s.view - 1, s', \cqc'}$
              holds true. Notice that the conditions~(\ref{eq:view-under}),
              (\ref{eq:view-eq}), and~(\ref{eq:block-no-above}) are not applicable in this case.
              Hence, we have condition~(\ref{eq:hv-eq}), and thus, $s'.high\_vote \eqmod \cqc'.vote$
              holds true. By the same reasoning, $\invOne{r.view - 1, s', \cqc}$ holds true. The
              conditions~(\ref{eq:view-under}) and (\ref{eq:view-eq}) are not applicable, and
              condition~(\ref{eq:block-no-above}) does not hold true, as
              $\cqc'.vote.block\_number = \cqc.vote.block\_number$. Hence, condition~(\ref{eq:hv-eq}) holds true for~$\cqc$, and thus $s'.high\_vote \eqmod \cqc.vote$. Hence, $\cqc'.vote = \cqc.vote$.
              As~$r$ sets its $high\_vote$ to $\cqc'.vote$ when processing~$\cqc'$ in line~\ref{line:set-high-vote}, we obtains $s.high\_vote \eqmod \cqc.vote$.
        \end{enumerate}
    \end{enumerate}

  This finishes the proof, as we have considered all the cases.
\end{proof}

\begin{lemma}\label{lem:timeout-chain}
  Consider an arbitrary execution of the protocol and a view~$v \ge 0$. Assume that:
  
  \begin{itemize}
    \item The invariant $\invOne{u, s_u, \cqc_u}$ holds for every
      view~$u < v$, commit quorum certificate~$\cqc_u$,
      and every state~$s_u$ of every correct replica~$r_u \in \signers{\cqc}$ such that
      $s_u.view \le u$ and $\cqc_u.vote.view \le u$ in the execution.
  
    \item The invariant~$\invTwo{u, \cqc_u, \tqc_u}$ holds for every
      view~$u < v$, every commit quorum certificate~$\cqc_u$, and timeout quorum
      certificate~$\tqc_u$ such that $view(\tqc_u) \le u$.
  \end{itemize}
  
  Then, $\invTwo{v, \cqc, \tqc}$ holds for the view~$v$, every commit quorum
  certificate~$\cqc$, and timeout quorum certificate~$\tqc$ such that $view(\tqc) \le v$.
\end{lemma}

\begin{proof}
    Fix a commit quorum certificate~$\cqc$ and a timeout quorum certificate~$\tqc$ such
    that $view(\tqc) \le v$.   
    If the conditions~(\ref{eq:next_block})--(\ref{eq:cqc_view_below}) do not hold true,
    the invariant $\invTwo{v, \cqc, \tqc}$ is trivially satisfied.
    Hence, we assume that these conditions are satisfied.

    Further, we notice that if $v > view(\tqc)$, our proof goal immediately follows
    from the invariant $\invTwo{v - 1, \cqc, \tqc}$, which holds by the lemma assumption.
    Hence, we assume that~$v = view(\tqc)$.

    Since $\cqc$ and $\tqc$ are valid quorum certificates, we apply
    Corollary~\ref{cor:commit-timeout-intersections} and then Corollary~\ref{cor:correct},
    to obtain the lower bound on their intersection with the set of the correct replicas
    $Q = (\signers{\cqc}\, \cap\, \signers{\tqc})\,\cap\, \Corr{}$:
\begin{equation}\label{eq:lemma2-subquorum}
        |Q| \ge n - 3f
    \end{equation}

    We proceed by considering two major cases.

\textit{(i) \underline{Case $v = view(\tqc) = \cqc.vote.view$.}}
  Since the correct replicas in~$Q$ send both commit votes (line~\ref{line:send-commit}) and
  timeout votes (line~\ref{line:send-timeout}) in view~$v$, they go through the phases \textit{Prepare},
  \textit{Commit}, and \textit{Timeout}. Fix an arbitrary replica~$r \in Q$. Replica~$r$ sends its commit
  vote~$cv$ in view~$v$ in line~\ref{line:send-commit}, replica~$r$ sets $high\_vote$ to $cv$. After that, 
  replica~$r$ sends its timeout vote~$tv$ with~$high\_vote = cv$ in line~\ref{line:send-timeout}.
  The above conclusion holds true for all replicas in~$Q$.
  From this and inequality~(\ref{eq:lemma2-subquorum}), we conclude that, when computing~$high\_vote(\tqc)$, the condition in line~\ref{line:q1} holds true.
  Therefore, there is at least one subquorum of~$n - 3f$ timeout votes.

  It remains to show that the condition in line~\ref{line:q2} does not hold true, when
  computing~$high\_vote(\tqc)$. On contrary, assume that there is another
  subquorum~$Q' \subseteq \signers{\tqc}$ with $Q' \cap Q = \emptyset$ and $|Q'| \ge n - 3f$.
  By the protocol assumption, $n > 5f$, and thus, $|Q'| > 2f$.
  As~$|\signers{\cqc}| \ge n - f$, by Corollary~\ref{cor:commit-intersections}, we have
  $|\,\signers{\cqc}\, \cap\, \Corr{}\,| \ge n - 2f$. Hence, at least $n - 2f$ correct replicas
  send a timeout vote that has~$cv$ as the high vote, which, by our assumption, is different from
  the high vote sent by the replicas in~$Q'$. As the total number of replicas is~$n$,
  we conclude that $|Q'| \le 2f$. We arrive at the contradiction with $|Q'| > 2f$.
  Therefore, there is no other subquorum of~$n - 3f$ timeout votes.

  We conclude that $high\_vote(\tqc) = cv = \cqc.vote$. This immediately gives us the proof
  goal~(\ref{eq:hv-eqmod}).

\textit{(ii) \underline{Case $v = view(\tqc) > \cqc.vote.view$.}} Consider an arbitrary correct
  replica~$q$ in view~$v$ that has signed both the timeout QC~$\tqc$ and the commit QC~$\cqc$,
  that is, $q \in Q$.
  First, we show that~$q$ sends a timeout vote~$tv$ in view~$v$ that has the following property:
\begin{align}
    tv.high\_vote \eqmod \cqc.vote \label{eq:tv-eqmod}
  \end{align}

  After proving~(\ref{eq:tv-eqmod}) for~$q \in Q$, we will show that the condition(\ref{eq:hv-eqmod})
  is satisfied.
  
  \underline{Proving~(\ref{eq:tv-eqmod})}. To avoid confusion, we use the
  notation~$\tqc_v$ for $\tqc$ in the rest of the proof. First,
  we notice that~$q$ must have received a
  timeout QC as a quorum certificate in the proposal handler (line~\ref{line:on_proposal}).
  Otherwise, it would have
  advanced the block number to~$\cqc.vote.block\_number$ in line~\ref{line:implied-next-block},
  which contradicts the assumption~(\ref{eq:next_block}). Let's call this timeout quorum
  certificate~$\tqc_{v-1}$. We notice that the invariant $\invTwo{v - 1, \cqc, \tqc_{v-1}}$
  holds by the lemma assumption. This gives us:
\begin{align}
    high\_vote(\tqc_{v-1}) \eqmod \cqc.vote \label{eq:hv-pred-v}
  \end{align}
  
  Further, to sign ~$\tqc_v$ in view~$v$, the replica~$q$ had to send a timeout
  vote in line~\ref{line:send-timeout}, which required $phase \ne \textit{Timeout}$.
  Let~$s$ be a local state of replica~$r$ just before the transition, in which~$r$ sent a
  timeout vote for the view~$v$. Note that $q$ sends a timeout vote with $high\_vote$ set
  to $s.high\_vote$. Consider two remaining cases on~$s.phase$:

  \begin{enumerate}
    \item $s.phase = \textit{Prepare}$. In this case, replica~$q$ set the value of $s.high\_vote$
      in view~$v - 1$. Hence, $r$ has sent a timeout vote~$tv$ with the field $high\_vote$ set to
      $high\_vote(\tqc_{v-1})$ in view~$v$. From this and~(\ref{eq:hv-pred-v}), we
      have the required condition~(\ref{eq:tv-eqmod}).
      
    \item $s.phase = \textit{Commit}$. In this case, replica~$q$ sets the value of $high\_vote$
      in view~$v$ when processing a proposal message in line~\ref{line:on_proposal}.
      As we have concluded above, $q$ receives $T_{v-1}$ as a quorum certificate.
      From~(\ref{eq:hv-pred-v}) and assumption~(\ref{eq:next_block}),
      we conclude that~$q$ executed line~\ref{line:implied-reproposal} in
      \texttt{get\_implied\_block}. As a result, $q$ sets its field $high\_vote$ to
      \texttt{<CommitVote, $T_{v-1}.view$ + 1, $\cqc.vote.block\_number$, $\cqc.vote.block\_hash$>}
      in line~\ref{line:create-commit}. Replica~$q$ uses this vote as $high\_vote$ when
      sending a timeout vote. The condition~(\ref{eq:tv-eqmod}) follows.
  \end{enumerate}

  We have shown that the condition~(\ref{eq:tv-eqmod}) holds true. Since we have fixed an
  arbitrary replica~$q \in Q$, and $|Q| \ge n - 3f$, by~(\ref{eq:lemma2-subquorum}),
  we conclude that when computing~$high\_vote(\tqc)$, the condition in line~\ref{line:q1}
  of \texttt{get\_implied\_block} holds true. That is, there is at least one subquorum
  of~$n - 3f$ timeout votes.

  It remains to show that the condition in line~\ref{line:q2} of \texttt{get\_implied\_block}
  does not hold true when computing~$high\_vote(\tqc)$, that is, there is no other subquorum
  of at least~$n - 3f$ timeout votes.
  
  To this end, fix an arbitrary replica~$r$ from~$\signers{\cqc} \cap \Corr$ and the local
  state~$t$ of~$r$ after sending its timeout vote for~$\tqc_{v - 1}$ in view~$v - 1$. Since
  $v > \cqc.vote.view$ and~$t.view = v - 1$,
  the invariant~$\invOne{v - 1, t, \cqc}$ holds true by the lemma assumption. The conditions
  (\ref{eq:view-under})--(\ref{eq:view-eq}) do not hold by our assumptions about the state~$t$.
  We conclude that either (\ref{eq:hv-eq}), or (\ref{eq:block-no-above}) hold for the state~$t$
  of replica~$r$. Notice that if the condition~(\ref{eq:block-no-above}) was true, then
  $s.high\_commit.vote.block\_number \ge \cqc.vote.block\_number$, and, as a result:
\begin{align}
  \tqc.high\_commit\_qc.vote.block\_number \ge \cqc.vote.block\_number
  \end{align}
  
  This is due to~$r$ having progressed to the next block and having $high\_commit\_qc$
  set to~$\cqc$ or a higher commit QC. This contradicts
  the assumption~(\ref{eq:next_block}). Hence, the condition~(\ref{eq:hv-eq}) holds true for~$r$.
  As a result, every correct replica from~$\signers{\cqc} \cap \Corr$ has its high vote set
  to~$\cqc.vote$. Recall that there are at least~$n - f$ replicas
  in~$\signers{\cqc}$, and, by Corollary~\ref{cor:correct}, at least~$n - 2f$ replicas in
  $\signers{\cqc} \cap \Corr$. Therefore, at most~$2f$ replicas can send timeout votes for~$\tqc$
  that contain $high\_vote$ different from $\cqc.vote$.

  As a result, there is only one subquorum when computing~$high\_vote(\tqc)$, and this subquorum
  contains timeout votes with $high\_vote \eqmod \cqc.vote$. Hence, $high\_vote(\tqc) \eqmod \cqc.vote$.
  This finishes the proof.
\end{proof}

\begin{lemma}\label{lem:cqc-agree}
  Consider an arbitrary execution of the protocol. Let $\cqc_1$ and $\cqc_2$ be two valid commit quorum
  certificates that are issued by replicas (correct or faulty).

  If the block numbers in their commit votes coincide:
\begin{equation}
    \cqc_1.vote.block\_number = \cqc_2.vote.block\_number \label{eq:cqc-same-number}
  \end{equation}
  
  Then, the block hashes also coincide:
\begin{equation}
    \cqc_1.vote.block\_hash = \cqc_2.vote.block\_hash \label{eq:cqc-same-hash}
  \end{equation}
\end{lemma}

\begin{proof}
  Consider two commit quorum certificates~$\cqc_1$ and $\cqc_2$ that
  satisfy~(\ref{eq:cqc-same-number}). Without loss of generality, assume that
  $\cqc_1.vote.view \leq \cqc_2.vote.view$.
  
  \underline{\textit{Case (i):}} $\cqc_1.vote.view = \cqc_2.vote.view$. We invoke
  Lemma~\ref{lem:cqc-one-vote}, to achieve our proof goal~(\ref{eq:cqc-same-hash}).

  \underline{\textit{Case (ii):}} $\cqc_1.vote.view < \cqc_2.vote.view$. By
  Corollary~\ref{cor:commit-intersections}
  and Corollary~\ref{cor:correct}, at least~$n - 3f$ correct replicas have signed~$\cqc_1$ and~$\cqc_2$.
  Since $n \ge 5f + 1$, by the protocol assumptions, there is at least one correct
  replica~$r \in \signers{\cqc_1} \cap \signers{\cqc_2} \cap \Corr$.

  Consider two main cases of choosing the block to propose or re-propose, according 
  to the definition of~\texttt{get\_implied\_block}:

  \begin{itemize}
    \item \underline{\textit{Case (ii.1):}} Replica~$r$ has received a timeout quorum certificate~$\tqc$
        for the view~$\cqc_2.vote.view - 1$ as a proposal justification in the view $\cqc_2.vote.view$.
   
        In this case, replica~$r$ invoked line~\ref{line:implied-reproposal} 
        in~\texttt{get\_implied\_block}. Otherwise, $r$ would increase the block number, either
        in line~\ref{line:implied-next-block} or in line~\ref{line:implied-next-block2}, as required
        by all other cases of~\texttt{get\_implied\_block}.
        By Lemma~\ref{lem:timeout-chain}, the invariant~$\invTwo{\cqc_2.vote.view, \cqc_1, \tqc}$
        holds true.
        
        to obtain the following:
\begin{equation*}
            high\_vote(\tqc) \eqmod \cqc_1.vote \label{eq:hv-eqmod1}
        \end{equation*}
        
        Hence, the replica $r$ returned $\cqc_1.vote.block\_hash$
        in line~\ref{line:implied-reproposal} when processing the proposal in view $\cqc_2.vote.view$.
        Further, replica $r$ created a new commit vote with
        this block hash in line~\ref{line:create-commit} and sent it in line~\ref{line:send-commit}.
        As all signers of~$\cqc_2$ have sent the same commit vote as~$r$, we immediately arrive
        at the proof goal~(\ref{eq:cqc-same-hash}).

    \item \underline{\textit{Case (ii.2):}} Replica~$r$ has received a commit quorum
        certificate~$\cqc_3$ for the view~$\cqc_2.vote.view - 1$ as a proposal justification in
        the view $\cqc_2.vote.view$. Hence, function \texttt{get\_implied\_block} in 
        line~\ref{line:implied-next-block} increases the block number. Therefore, we have:
\begin{align}
            \cqc_3.vote.block\_number = \cqc_1.vote.block\_number - 1 = \cqc_2.vote.block\_number - 1
            \label{eq:c3-block-number}
        \end{align}

        Now we consider the relation between~$\cqc_1$ and~$\cqc_3$.
        By Corollary~\ref{cor:commit-intersections} and Corollary~\ref{cor:correct}, at least~$n - 3f$ correct replicas have signed~$\cqc_1$ and~$\cqc_3$. 
        Since $n \ge 5f + 1$, by the protocol assumptions, there is at least one correct
        replica~$r' \in \signers{\cqc_1} \cap \signers{\cqc_3} \cap \Corr$.

        It remains to consider three cases on the relation between the views of~$\cqc_1$ and~$\cqc_3$:

        \begin{itemize}
            \item \underline{\textit{Case (ii.2.a):}} $\cqc_1.vote.view = \cqc_3.vote.view$.
            In this case, we apply Lemma~\ref{lem:cqc-one-vote} to obtain~$\cqc_1.vote =\cqc_3.vote$.
            This immediately gives us a contradiction with equation~(\ref{eq:c3-block-number}).
            Thus, this case is impossible.

            \item \underline{\textit{Case (ii.2.b):}} $\cqc_1.vote.view > \cqc_3.vote.view$.
            In this case, we have:
\begin{align}
              \cqc_1.vote.view &< \cqc_2.vote.view, \text{ by the assumption of Case (ii)}
                \label{eq:c2-view-lt-c3-view} \\
              \cqc_2.vote.view &= \cqc_3.vote.view + 1 \label{eq:c2-view-succ-c3-view}
            \end{align}

            By substituting the value of~$\cqc_2.vote.view$ from equation~(\ref{eq:c2-view-succ-c3-view})
            into equation~(\ref{eq:c2-view-lt-c3-view}), we
            obtain~$\cqc_1.vote.view < \cqc_3.vote.view + 1$. This is in immediate contradiction
            with the case assumption~$\cqc_1.vote.view > \cqc_3.vote.view$. Thus, this case is impossible.

            \item \underline{\textit{Case (ii.2.c):}} $\cqc_1.vote.view < \cqc_3.vote.view$.
            Consider the state~$s'$ of replica~$r'$ right after it has sent its commit vote for~$\cqc_3$
            in view~$\cqc_3.vote.view$ in line~\ref{line:send-commit}. By applying Lemma~\ref{lem:inv1},
            we obtain validity of~$\invTwo{\cqc_3.vote.view, s', \cqc_1}$.
            Hence, one of the conditions~(\ref{eq:view-under})--(\ref{eq:block-no-above}) must
            be satisfied. The conditions~(\ref{eq:view-under})--(\ref{eq:view-eq}) are not
            satisfied, since~$\cqc_3.vote.view > \cqc_1.vote.view$ (by the case assumption).
            The conditions~(\ref{eq:hv-eq})--(\ref{eq:block-no-above}) are not satisfied,
            due to equation~(\ref{eq:c3-block-number}). Thus, we arrive at a contradiction.
        \end{itemize}

        As a result, the case of replica~$r$ using a commit QC as a justification for sending
        a commit vote for the commit quorum certificate~$\cqc_2$, that is \textit{Case (ii.2)},
        is impossible.
  \end{itemize}

  Finally, we have shown that \textit{Case~(i)} and \textit{Case~(ii.1)} give us the proof
  goal~(\ref{eq:cqc-same-hash}), whereas \textit{Case~(ii.2)} is impossible. This finishes the proof.
\end{proof}

Having Lemma~\ref{lem:cqc-agree}, it is easy to show that the protocol satisfies the agreement property:

\begin{theorem}[Agreement]\label{thm:agreement}
  Consider an arbitrary execution of the protocol. Let $r_1$ and $r_2$ be two correct replicas.
  If the replicas $r_1$ and $r_2$ commit blocks $b_1$ and $b_2$ for a block number $n \in \mathbb{N}_0$,
  respectively, then $b_1 = b_2$.
\end{theorem}

\begin{proof}
  The only way for $r_1$ and $r_2$ to commit a block is by executing line~\ref{line:commit}
  in Listing~\ref{lst:replica}, in a call to~\texttt{process\_commit\_qc}. Let $\cqc_1$ and $\cqc_2$ be the commit
  quorum certificates that are passed to~\texttt{process\_commit\_qc} by~$r_1$ and~$r_2$, respectively.
  Since blocks $b_1$ and $b_2$ have the same block number~$n$,
  we invoke Lemma~\ref{lem:cqc-agree}. As a result, the blocks $b_1$ and $b_2$
  have the same hash. Hence, we conclude that $b_1$ and $b_2$ are identical.
\end{proof}

The second important safety property is validity, which essentially holds by algorithm design:

\begin{theorem}[Validity]\label{thm:validity}
  Consider an arbitrary execution of the protocol. If a correct
  replica~$r$ commits a block~$b$, this block~$b$ passes block verification,
  and it was proposed earlier.
\end{theorem}

\begin{proof}
  The only possibility for a correct replica to cache a block (not only its 
  hash) is by processing a proposal and verifying it in
  line~\ref{line:on-proposal-verify-block}. Further, a block is committed  in 
  line~\ref{line:process-commit-end} under the condition that the block
  was previously cached by a replica. Hence, validity holds true.
\end{proof}

\subsection{Liveness}

To reason about liveness, we need additional assumptions about the communication model.

\subsubsection{Assumptions about time and message delivery}\label{sec:partial-sync}

Similar to HotStuff-2~\cite{malkhi23}, we assume partial synchrony~\cite{dwork1988consensus}.
More precisely, we assume that in every protocol execution, there is a global-stabilization-time
(GST) such that after \gst{} all messages from correct replicas to correct replicas arrive within
the bound of~$\delay$ time units. \gst{} is unknown to the replicas and thus cannot be used in the protocol.
In practice, the periods of asynchrony and synchrony may alternate. However, we follow the standard
assumptions of \gst{}, to keep the proofs simple.

It is important to know the guarantees on delivery of the messages that were sent before~\gst{}.
In ChonkyBFT, we assume that every replica periodically retransmits the latest message of each type
that it had sent previously, that is, the latest messages of types
\textit{SignedCommitVote}, \textit{SignedTimeoutVote}, and \textit{NewView}. Hence, after reaching
\gst{}, each correct replica receives the latest message of each type sent by every correct
replica before \gst{}. We assume that all such messages are delivered by the time~$\gst + \delay$ at latest.
In the following proofs, we refer to time points such as~$t$ and $t + \delay$. Hence, our proofs assume
the global view of the system. The replicas do not have access to a global clock,
but only can measure timeouts such as~$\delay$, which assumes that replicas' clocks advance at the same rate.

\subsubsection{Block synchrony}\label{sec:block-sync}

We notice that a correct replica in a view~$v$ can catch up with faster correct
replicas in a view $v' > v$ by receiving a \textit{NewView} message, \textit{SignedCommitVote}
from~$n - f$ replicas, or~\textit{SignedTimeoutVote} from~$n - f$ replicas. By doing so,
a correct replica may skip receiving a full block that was proposed in the view~$v$.
Moreover, by catching up with the faster replicas, such a slower replica may introduce
a gap in its list~\texttt{committed\_blocks}.

ChonkyBFT replicas run an additional protocol to fetch blocks that were proposed, cached,
and committed by other replicas. We do not detail this protocol, but simply
assume that there is a constant~$\tsync \ge \delay$ and the synchronization protocol that
satisfy the following:

\begin{itemize}
  \item Whenever a block is proposed at time~$t \ge \gst{}$, this block is delivered
        to every correct replica by the time~$t + \tsync$ and cached by it.

  \item Whenever a correct replica commits a block with a commit QC at time~$t$,
        by the time $max(t, \gst) + \tsync{}$,
        this block and the quorum certificate are delivered to every correct replica,
        and every correct replica adds the block to its list~\texttt{committed\_blocks}.
        (It is safe to do so by Theorem~\ref{thm:agreement}).
\end{itemize}

\subsubsection{Proofs of Liveness}

To show liveness of ChonkyBFT, we have to demonstrate that the property of Progress
(Section~\ref{sec:properties}) holds true. To this end, we first prove
Lemmata~\ref{lem:sync-views}--\ref{lem:faulty-leader-progress}.

The following lemma establishes view synchronization among the correct replicas:

\begin{lemma}\label{lem:sync-views}
  Let $v$ be the view a correct replica~$r$ at time~$t \ge \gst$.
  By the time~$t + \delay$, every correct replica is in the view~$v$, or in a greater
  view.
\end{lemma}

\begin{proof}
  Fix a correct replica~$r$, and let~$v$ be its view at time~$t$. If $v = 0$, then the proof
  goal follows immediately, since all correct replicas are in view~0 or great.
  Consider the case of $v > 0$. In this case, replica~$r$ sends a 
  message~$\msg{\textit{NewView}, qc}$ in
  line~\ref{line:send-new-view}. Since~$r$ is in the view~$v$ at $t \ge \gst$, the 
  message~$\msg{\textit{NewView}, qc}$ is either
  the latest message sent by~$r$ before~\gst, or it is a message sent by~$r$ after~\gst.
  By our assumption of partial synchrony, in both cases, the
  message~$\msg{\textit{NewView}, qc}$ must be delivered to all correct replicas by the
  time~$t + \delay$.

  Consider an arbitrary replica~$r'$. Replica~$r'$ receives the
  message~$\msg{\textit{NewView}, qc}$ in line~\ref{line:on-new-view-start}
  at time $t' \in [t, t + \delay]$. If~$r'.view < v$ at~$t'$, replica~$r'$
  immediately switches to the view~$v$, as prescribed in line~\ref{line:on-new-view-end}.
  This finishes the proof.  
\end{proof}

The following lemma establishes block progress for the case of a correct leader:

\begin{lemma}\label{lem:same-view-commit}
  Let $v_{max}$ be the maximum view among the views of the correct replicas at a
  time point~$t \ge \gst{}$, and no correct replica is in view~$v_{max}$ at every earlier
  time point $t' < t$. Assume that the following conditions are satisfied:
\begin{enumerate}
      \item The leader of view~$v$ is correct, and \label{item:leader-is-correct}
      
      \item The timeout value is large enough: $\timeout > \tsync + 2\delay$.  \label{item:timeout-is-large}
  \end{enumerate}

  Then, by the time~$t + \tsync + 2\delay$, the following holds true:

  \begin{enumerate}
      \item All correct replicas reside in views no less than~$v_{max}+1$,
        and \label{item:next-view}
        
      \item All correct replicas send a \textit{NewView} message with a commit
            QC for the same block number as a justification. \label{item:send-new-view}
  \end{enumerate}
\end{lemma}

\begin{proof}

  By Lemma~\ref{lem:sync-views}, no later than by~$t + \delay$, all correct replicas switch
  to the view~$v_{max}$, including the leader of the view~$v_{max}$. This leader
  is correct by assumption~(\ref{item:leader-is-correct}). Moreover, by the assumption of
  Section~\ref{sec:block-sync}, the condition in line~\ref{line:proposer-has-blocks}
  holds true by the time~$t + \tsync$, since $\tsync \ge \delay$. Thus, the
  replica~$\leader{v_{max}}$ sends a single proposal by the time~$t + \tsync$, say,
  $\left<\text{Proposal}, b, qc\right>$.
  
  By the time~$t + \tsync + \delay$, all correct replicas receive the proposal
  $\left<\text{Proposal}, b, qc\right>$. Again, by the assumption of
  Section~\ref{sec:block-sync}, 
  the condition in line~\ref{line:on-proposal-has-blocks}
  holds true by the time~$t + \tsync$, that is, every correct replica has received the blocks
  for the smaller block numbers. Hence, all correct replicas send their commit votes
  in line~\ref{line:send-commit} by the time~$t + \tsync + \delay$. These commit votes
  are received by the time~$t + \tsync + 2\delay$. Hence, all correct replicas switch
  to the view~$v_{max} + 1$, either by receiving commit votes from~$n - f$ replicas
  in lines~\ref{line:on-commit-begin}--\ref{line:on-commit-end}, or by
  receiving a \textit{NewView} message for the view~$v_{max}$ in
  lines~\ref{line:on-new-view-start}--\ref{line:on-new-view-end}.
  (We discuss below why it is impossible for the correct replicas to receive timeout votes
  from~$n - f$ replicas.) This proves the goal~(\ref{item:next-view}).
  
  Moreover, since all correct replicas process the same proposal,
  they process a commit QC for the same block number, when switching to the view~$v_{max}+1$
  in lines~\ref{line:start-new-view-start}--\ref{line:start-new-view-end}.
  This proves the goal~(\ref{item:send-new-view}).

  It remains to show that the correct replicas cannot switch to the view~$v_{max}+1$
  by receiving timeout votes in lines~\ref{line:on-timeout-start}--\ref{line:on-timeout-end}.
  To this end, notice that the correct replicas send their timeout votes no earlier than
  after their timers expire in
  lines~\ref{line:on-send-timeout-start}--\ref{line:on-send-timeout-end}.
  The earliest such expiration happens at time~$t + \timeout$, since by the lemma's assumption,
  the time point~$t$ is the earliest point when at least one correct replica switches to
  the view~$v_{max}$. By the assumption~(\ref{item:timeout-is-large}), we
  have~$\timeout > \tsync + 2\delay$. Therefore, the correct replicas may send timeout
  votes in view~$v_{max}$ at a time point $t' > t + \tsync + 2\delay$, or later.
  However, as we showed above, by the time $t + \tsync + 2\delay$, the correct replicas
  can switch to the view~$v_{max} + 1$. The faulty replicas alone can only send up to~$f$
  timeout votes from the unique senders, which is insufficient to collect~$n - f$ timeout
  votes. Hence, no correct replica switches to the view~$v_{max}+1$ via a timeout QC.
\end{proof}

The following lemma establishes view progress even in the case of a faulty leader, so
a correct replica has a chance of becoming a view leader:

\begin{lemma}\label{lem:faulty-leader-progress}
  Let $v_{max}$ be the maximum among the views of the correct replicas at a time
  point~$t \ge \gst{}$. By the time~$t + \timeout + 2\delay$, all correct replicas reside in
  views no less than~$v_{max} + 1$.
\end{lemma}

\begin{proof}
  A correct replica may switch to the next view on one of the following events:

  \begin{enumerate}
      \item By receiving commit votes from at least $n - f$ replicas in
            lines~\ref{line:on-commit-begin}--\ref{line:on-commit-end}.

      \item By receiving a \textit{NewView} message in
            lines~\ref{line:on-new-view-start}--\ref{line:on-new-view-end}.

      \item By receiving timeout votes from at least $n - f$ replicas in
            lines~\ref{line:on-timeout-start}--\ref{line:on-timeout-end}.
  \end{enumerate}

  The key to the proof is to show that it is impossible for the faulty replicas
  to keep the correct replicas in different views, e.g., by forming commit quorums
  with one subset of the correct replicas, so the other subset of the correct
  replicas cannot build a timeout quorum.

  To this end, we consider the three cases:

  \underline{\textit{Case (i):}} At least one correct replica switches to the
  view~$v_{max}+1$ by receiving commit votes from $n - f$ replicas in
  lines~\ref{line:on-commit-begin}--\ref{line:on-commit-end} at
  time~$t' \in [t, t + \timeout + \delay)$. In this case, invoking Lemma~\ref{lem:sync-views},
  we find that every correct replica switches to the view~$v_{max} + 1$ (or a larger view) by 
  $t' + \delay$, that is, by the time~$t + \timeout + 2\delay$.

  \underline{\textit{Case (ii):}} A correct replica receives a \textit{NewView}
  message with a valid commit QC or a valid timeout QC built in the view~$v_{max}$,
  and this message is received at time~$t' \in [t, t + \timeout + \delay)$.
  For instance, this case is possible when a faulty replica sends proposals to a subset of 
  correct replicas and builds a commit QC by adding the commit votes of the faulty replicas.
  As in Case~(i), invoking Lemma~\ref{lem:sync-views},
  we find that every correct replica switches to the view~$v_{max} + 1$ (or a larger view) by 
  time $t' + \delay$, that is, by the time~$t + \timeout + 2\delay$.

  \underline{\textit{Case (iii):}} Neither of the cases~(i) and~(ii) hold true.
  By Lemma~\ref{lem:sync-views}, all correct replicas switch to the view~$v_{max}$ by
  $t + \delay$. When switching to the view~$v_{max}$, all correct replicas set their
  timers to~$\timeout$ in line~\ref{line:reset-timer}. Since~$t \ge \gst{}$, by the
  time~$t + \delay + \timeout$, every correct replica sends its timeout vote for the
  view~$v_{max}$ in lines~\ref{line:on-send-timeout-start}--\ref{line:on-send-timeout-end}.
  Note that the correct replicas are still in the view~$v_{max}$, as the cases~(i) and~(ii)
  do not hold true by the assumption. From~$t \ge \gst{}$, we conclude that all correct
  replicas receive their timeout votes from all correct replicas by~$t + \timeout + 2\delay$.
  As a result, all correct replicas switch to the view~$v_{max} + 1$ in
  lines~\ref{line:on-timeout-start}--\ref{line:on-timeout-end}.

  We have shown the proof goal for all three cases.
\end{proof}

Finally, we are in a position to show Progress of ChonkyBFT.

\begin{theorem}[Progress]\label{thm:progress}
    Consider an execution of ChonkyBFT. Assume that the following conditions are satisfied:

    \begin{itemize}
        \item The leader function is fair, that is, for every replica~$r \in \Corr$
        and every view~$v \in \NatZero{}$, there is a view~$v' \ge v$ with $r = \leader{v'}$.

        \item The timeout value is large enough: $\timeout > \tsync + 2\delay$.

        \item The assumptions of partial synchrony (Section~\ref{sec:partial-sync}) and block
        synchrony (Section~\ref{sec:block-sync}) hold true.
    \end{itemize}
\end{theorem}

\begin{proof}
  We show for an arbitrary~$n \ge 0$, that if the correct replicas have committed~$n$ blocks,
  they eventually commit a block with the number~$n+1$. Let~$v$ be the smallest view, in
  which the correct replicas have committed~$n$ blocks. Since the leader function is fair,
  there is a view~$v' \ge v$ such that $\leader{v'}$ is correct; let~$v'$ be the minimal
  such view. By applying Lemma~\ref{lem:faulty-leader-progress} to the
  views~$v, v+1, \dots, v'$, we show that the correct replicas eventually switch to the 
  view~$v'$. Now we apply Lemma~\ref{lem:same-view-commit} to obtain that all correct
  replicas eventually have $n+1$ committed blocks in view~$v'$.
\end{proof}

\section{Formal Specification and Model Checking}\label{sec:formal-modeling}

Listings~\ref{lst:defs}--\ref{lst:replica2} are written in pseudocode. This pseudocode
is optimized for being concise, yet sufficiently formal, so the readers would be able
to understand the protocol mechanics and to be able to follow the mathematical proofs.
While this approach is customary in the distributed computing literature, it has a few
drawbacks:

\begin{enumerate}
    \item \textbf{Pseudocode cannot be statically analyzed}. As a result, it may contain
    errors, starting with typos and simple type errors (such as missing record fields)
    and continuing with logical errors.

    \item \textbf{Pseudocode cannot be executed}. This complicates the discussions
    about happy execution paths and potential issues. There is always a chance
    that the parties understand pseudocode differently.

    \item \textbf{Pseudocode cannot be checked by a machine}. Related to the issue of static
    analysis and execution, it is impossible to automatically analyze the key properties
    of a protocol written in pseudocode. For example, the protocol designers are mostly
    interested in whether their protocol is safe and live.
\end{enumerate}

In case the protocol has been implemented, the above points are less critical, as all
parties can compare the pseudocode against the implementation, even though the implementation
is bound to contain details that are less relevant for understanding of the protocol.
However, if the protocol is still in the design phase, it is crucial to have a machine-checkable
specification. This was the case for ChonkyBFT. We started with the informal pseudocode specification
of ChonkyBFT that resembles Rust~\cite{ChonkyInformal}. To refine this specification and make
it machine-checkable, we wrote a formal specification in Quint~\cite{Quint} and analyzed it
with the Quint simulator, as well as with the symbolic model checker
Apalache~\cite{Apalache, KKT19, KonnovKM22}.
In the course of writing the formal specification, we have indeed found a number of
issues in the pseudocode at the surface level such as typos, missing definitions, and type errors.
Moreover, the tools have found several paths that could lead to an attack in the implementation,
which were easy to fix.
This is not really surprising, as similar experience was reported multiple times when applying
formal methods such as~\tlap{} to designs of distributed systems, e.g., 
see~\cite{Newcombe14}.
The details of this work can be found in the technical blog post~\cite{ChonkyBlogpost}.

Why did we choose Quint? It the recent years, it became common to specify distributed
protocols in~\tlap{}~\cite{Lamport2002} and analyze them with the model checkers TLC and Apalache.
For instance, there are~\tlap{} specifications of Paxos~\cite{tla-paxos}, Raft~\cite{tla-raft}, 
Tendermint~\cite{TendermintSpec2020}, and, most recently, TetraBFT~\cite{tla-tetrabft}.
A large corpus of~\tlap{} specifications can be found at~\cite{tla-examples}.

Notwithstanding the success of~\tlap{}, the language is known to have a steep learning curve.
This is probably the reason for why beginners prefer to use PlusCal~\cite{pluscal} instead
of~\tlap{}. Even though PlusCal looks more like a programming language, it is significantly
more restricted than~\tlap{}. In our own experience, PlusCal is a good fit for expressing
concurrent algorithms such as mutual exclusion and nonblocking multithreaded algorithms,
but it introduces unnecessary boilerplate when expressing distributed consensus algorithms.
In contrast to PlusCal, Quint borrows syntax from the modern programming languages such as Scala
and Rust, yet it keeps very close correspondence with the logic of~\tlap{}.
In addition to the more familiar syntax, Quint has a built-in type checker, which
enforces the type system that was developed for the model checker Apalache.
On top of that, Quint has a built-in simulator that helps the engineers experimenting
with the protocol without urgency to use more advanced formal verification tools.

\subsection{Specifying ChonkyBFT in Quint}

In the first phase of the specification efforts, we manually translated the pseudo
code specification~\cite{ChonkyInformal} to the Quint specification~\cite{ChonkyQuint}.
Listing~\ref{lst:quint-on-commit} shows a small part of the specification that
processes signed commit votes. The resulting specification was structurally similar to the
original Rust-like pseudocode. However, there are noticeable differences:

\begin{itemize}
    \item Following the methodology of~\tlap{}, we model the protocol as a state machine.
    Hence, we declare state variables for the local states of the correct replicas, as well
    as for the global view of the system. For example, the states of the replicas are tracked
    with the map \texttt{replica\_state}, whereas the broadcasted commit votes are tracked
    with the set \texttt{store\_signed\_commit}. As a result, we have to propagate the values
    of all state variables, even if they are not changed.

    \item We model the potential behavior of the faulty replica by injecting messages
    that could be sent by the faulty replicas at any point in time.

    \item We model timeouts as non-deterministic events. While this is sufficient for reasoning
    about protocol safety, to reason about liveness, we would need a more refined specification.
\end{itemize}

\begin{lstlisting}[float, columns=fullflexible, caption={The handler of commit votes in Quint},
  language=quint, label=lst:quint-on-commit]
action on_commit(id: ReplicaId, signed_vote: SignedCommitVote): bool = all {
  val self = replica_state.get(id)
  all {
    // if the vote isn't current, just ignore it.
    signed_vote.vote.view >= self.view,
    // Check that the signed vote is valid.
    signed_commit_vote_verify(signed_vote),
    // Store the vote. We will never store duplicate (same view and sender) votes.
    is_new_signed_commit_vote_for_replica(id, signed_vote),
    val new_store = store_signed_commit.get(id).union(Set(signed_vote))
    all {
      store_signed_commit' = store_signed_commit.set(id, new_store),
      // Check if we now have a commit QC for this view.
      val qc_opt = get_commit_qc(new_store, signed_vote)
      match (qc_opt) {
        | None => all {
          replica_state' = replica_state, msgs_new_view' = msgs_new_view,
          ghost_justifications' = ghost_justifications,
        }
        | Some(qc) => {
          val self1 = self.process_commit_qc(Some(qc))
          start_new_view(id, self1, signed_vote.vote.view + 1, Commit(qc))
        }
      },
      msgs_signed_commit' = msgs_signed_commit, msgs_proposal' = msgs_proposal,
      msgs_signed_timeout' = msgs_signed_timeout,
      store_signed_timeout' = store_signed_timeout,
      proposer_view' = proposer_view, ghost_step' = OnCommitStep(id),
} } }
\end{lstlisting}

In addition to the specification itself, we have written test runs in Quint.
Listing~\ref{lst:test-run} shows a shortened version of such a test run. These runs are
indispensable in making sure that the expected protocol scenarios are reproducible by
the specification. Further, these runs work as basic tests, making sure that no errors were
introduced in the specification itself. Moreover, we wrote test runs that faulty replicas
may indeed break the protocol safety, when the protocol assumptions do not hold true, e.g.,
when $n < 5f + 1$.

\begin{lstlisting}[float, columns=fullflexible, caption={A test run in Quint},
  language=quint, label=lst:test-run]
// the leaders propose a block and all replicas commit that block
run replicas_normal_case_Test = {
  init_view_1_with_leader(Map(0 -> "n0", 1 -> "n0", 2 -> "n1", 3 -> "n2", 4 -> "n3"))
  // the leader proposes
  .then(all { proposer_step("n0", "val_b0"), unchanged_leader_replica })
  .expect(all_invariants)
  // replicas process the propose message
  .then(all_replicas_get_propose("val_b0", 0, "n0", Timeout(init_timeout_qc)))
  .expect(all_invariants)
  // n0 process the commit messages
  .then(replica_commits("n0", "val_b0", 0, 1))
  // [omitted for brevity...]
  // n4 process the commit messages
  .then(replica_commits("n4", "val_b2", 2, 3))
  .expect(all_invariants)
  .then(all {
    assert(replica_committed("n0") == [ "val_b0", "val_b1", "val_b2" ]),
    // ...
  }) }
\end{lstlisting}

Our specification includes 23 state invariants, starting with basic invariants 
that ensure integrity of the data structures and ending with the crucial safety
invariants such as the property of Agreement. Two such invariants are shown in
Listing~\ref{lst:quint-invariants}.

\begin{lstlisting}[float, columns=fullflexible,
  caption={Examples of state invariants for ChonkyBFT in Quint},
  language=quint, label=lst:quint-invariants]
// a correct proposer should not send different proposals in the same view
val no_proposal_equivocation_inv =
  tuples(msgs_proposal, msgs_proposal).forall(((m1, m2)) => or {
    not(m1.justification.view() == m2.justification.view() and m1.sig == m2.sig),
    FAULTY.exists(id => sig_of_id(id) == m1.sig),
    m1.block == m2.block,
  })

// no two correct replicas disagree on the committed blocks
val agreement_inv = tuples(CORRECT, CORRECT).forall(((id1, id2)) => {
  val blocks1 = replica_state.get(id1).committed_blocks
  val blocks2 = replica_state.get(id2).committed_blocks
  or {
    // ignore this case, as there is a symmetric one
    blocks1.length() > blocks2.length(),
    // the shorter sequence is a prefix of the longer one
    blocks1.indices().forall(i => blocks1[i].block == blocks2[i].block)
  }
})
\end{lstlisting}

\subsection{Simulation and Bounded Model Checking}\label{sec:sim-and-mc}

To quickly debug our specifications, we were periodically checking the state invariants
with the Quint simulator. This simulator randomly picks values from the sets specified
in the Quint actions. The simulator explores executions and evaluates state invariants,
similar to stateful property-based testing. This allowed us to find examples of executions
as well as shallow bugs in a matter of minutes.

Randomized simulation can only produce examples of invariant violation, if it
finds one, but, by no means, it can give us any guarantee of whether a state invariant 
holds true for all reachable protocol states. To this end, we also conducted experiments 
with the symbolic model checker Apalache. The model checker supports two exploration
strategies:

\begin{itemize}
    \item \textbf{Bounded model checking}. The model checker
      checks all executions up to~$k$ steps. Further, if it does not find an invariant
      violation, we have a hard guarantee that this invariant holds true of the specification
      on all states that are reachable within~$k$ steps.

    \item \textbf{Randomized symbolic simulation}. The model checker constructs a symbolic
      path by picking actions at random. It may miss invariant violation, when it requires
      a specific unexplored path. However, the state coverage of randomized symbolic execution
      is much higher than the state coverage of random simulation.
\end{itemize}

As it often happens with specifications of sophisticated fault-tolerant distributed
algorithms, bounded model checking was useful in finding invariant violations up to
about~10 steps, but it dramatically slowed down afterwards. Hence, we mostly used 
randomized symbolic simulation, as our main tool for making sure that the invariants were
not violated. To this end, we were running dozens of experiments with GNU
parallel~\cite{tange_2024_11247979}, to symbolically explore protocol executions up
to 25--30 steps. Further details are to be found in the blog post~\cite{ChonkyBlogpost}.

\subsection{Showing Safety by Constructing an Inductive Invariant}

As we have discussed in Section~\ref{sec:sim-and-mc}, symbolic simulation,
bounded model checking, and randomized symbolic execution increase our confidence in
protocol correctness to various degrees, but none of these techniques alone give us a
guarantee of safety for arbitrary protocol executions.

Since we wanted to obtain guarantees for executions of arbitrary lengths, we have employed
another technique supported by Apalache, namely, the technique of inductive invariants.
To use this technique, we have to find a state predicate~$\mathit{IndInv}$ that captures all
reachable states of the protocol. Further, we run the model checker to show the three conditions:

\begin{align}
  \mathit{Init} & \Rightarrow \mathit{IndInv} \label{eq:indinv-init}\\
  \mathit{IndInv} \land \mathit{Step} & \Rightarrow \mathit{IndInv}' \label{eq:indinv-step}\\
  \mathit{IndInv} & \Rightarrow \mathit{Agreement} \label{eq:indinv-agreement}
\end{align}

Apalache can check the conditions~(\ref{eq:indinv-init})--(\ref{eq:indinv-agreement}) for bounded
state spaces. However, it cannot automatically construct~$\mathit{IndInv}$. In theory, the
problem of automatically constructing such predicates is notoriously difficult. In practice,
we construct~$\mathit{IndInv}$ as a conjunction of a large number of small lemmas, each
of the lemmas capturing a single aspect of the invariant. For example, Listing~\ref{lst:lemma8}
shows one of the lemmas that we have written.

\begin{lstlisting}[float, columns=fullflexible,
  caption={One of the lemmas in the inductive invariant of ChonkyBFT},
  language=quint, label=lst:lemma8]
val lemma8_no_commits_in_future = {
  r::msgs_signed_commit.forall(m => or {
    IFAULTY.contains(m.sig),
    val self = r::replica_state.get(m.sig)
    or {
      m.vote.view < self.view,
      m.vote.view == self.view and self.phase != PhasePrepare,
    }
  })
}
\end{lstlisting}

We started with a small set of lemmas that captured the basic properties such as:

\begin{itemize}
    \item No equivocation by the correct replicas,
    \item Inability of the faulty replicas to sign messages on behalf of the correct replicas,
    \item Inability of the correct replicas to send messages from their own future,
    \item Support of commit and timeout quorum certificates with necessary messages.
\end{itemize}

Then, by iteratively running the model checker, we were checking the
conditions~(\ref{eq:indinv-init})--(\ref{eq:indinv-agreement}). As soon as Apalache was
reporting a counterexample, we were using it as the input to refine the existing lemmas,
or to introduce new lemmas.

The final inductive invariant consists of 18 such lemmas. Since it takes the model checker
significant time to check the condition~(\ref{eq:indinv-step}), we automatically decompose
the lemmas into smaller pieces and check them independently. For example, when a lemma has
the shape~$A_1 \land \dots \land A_k$, the conjuncts $A_1, \dots, A_k$ are checked as parallel
jobs. In the final iterations, the experiments consisted of 199 jobs.
To make use of this parallelism, we ran the experiments on a machine that is equipped with 
an AMD EPYC 7401P 24-Core Processor and 256G RAM.

\paragraph{Challenges.} When we started to construct the inductive invariant, we 
encountered
several difficulties. Some of them were typical for applying symbolic model checking to
this task, whereas some of the challenges were new. We briefly discuss these challenges,
to highlight the potential for further improvement. The scope of such improvements
typically belongs to long-term academic research rather than industrial engagements.

\textit{Restricting the parameter space}. Apalache encodes the model checking
checking problem in the quantifier-free fragment of SMT. Hence, we have to restrict
the input parameters. For example, we restrict the replica set to six replicas,
at most one of them being faulty. In addition, we restrict the set of views to three
elements.

\textit{Bounding the size of the initial data structures}. When the model checker
starts symbolic exploration from an initial state, the data structures such as the
sets of sent and received messages are quite small (usually empty). In the course of
symbolic exploration, these data structures grow. In the case of an inductive invariant,
we have to start in an arbitrary state that satisfies the invariant. In an arbitrary
state, the sets of messages can be astronomically large. For example, if we have
6 replicas, up to 3 views, up to 2 block numbers, and up to 3 block hashes, the
set of signed commit votes can have up to 108 elements. Under the same
restrictions, the set of proposals that use commit QC as a justification is as large
as 248 thousand elements. Since timeout votes carry commit votes and commit QCs,
we quickly arrive at extremely large sets.

Although potential sets of sent messages are extremely large, it is clear that
a single replica does not need to receive thousands of messages to take a single step.
For example, when processing a proposal, a replica needs a single proposal message.
Given this intuition, it may seem that the right approach is to select tiny sets of
arbitrary elements, e.g., a singleton set of proposals, a six-element set of commit
messages, and a six-element set of timeout messages. Even though this approach is
appealing, it would be insufficient, as the inductive invariant involves reasoning that
requires every new message to be supported by the messages received in the previous views.
Hence, these small subsets must be large enough to build such chains of reasoning.
Thus, we restrict the cardinalities of the message sets, given the specification 
parameters. Similar pragmatic approach to verification is employed in
Alloy~\cite{jackson2012software}, which is built around the \emph{small scope hypothesis},
that is, if a specification has an issue, this issue should be detectable on a small
input set.

\textit{Dealing with astronomic sizes of nested data structures}. Further, many data
structures in ChonkyBFT carry quorum certificates. This significantly simplifies
the manual proof and makes the algorithm easier to understand. However, Apalache constructs
symbolic constraints for the data structures such as the sets of messages and the replica
states. The potential sets of quorum certificates can grow very fast, as they contain
subsets of replicas. Moreover, timeout QCs refer to high commit QCs. To avoid this
kind of combinatorial explosion, we introduce small sets of timeout QCs and commit
QCs. As a result, all data structures refer to the QCs via their integer identifiers.
This approach significantly reduces the number of SMT constraints.

\section{Evaluation}

In this section, we evaluate the performance of ChonkyBFT consensus protocol through a series of experiments. We conducted tests with varying numbers of validators and block sizes to assess scalability, resource utilization, and latency. Only consensus throughput has been tested (i.e., block synchronization performance was out of scope). All experiments are performed under failure-free conditions.

\subsection{Experimental Setup}  
Each node was deployed as an \texttt{e2-highcpu-8} virtual machine (VM) on Google Cloud Platform (GCP) with a network bandwidth of \texttt{1 Gb/s egress}.

The nodes were distributed across all GCP regions in a round-robin manner, with the order of regions randomized. This ensured that the number of nodes per region differed by at most one, providing an even geographic distribution of nodes across the globe, even for smaller configurations.

To optimize network utilization, the system variable \texttt{net.ipv4.tcp\_slow\_start\_after\_idle} was set to 0 on all VMs. This adjustment was necessary because leader proposals, which dominate consensus network traffic, produce large bursts of data separated by periods of inactivity exceeding 10 seconds. The default TCP behavior under these conditions resulted in suboptimal performance.

We conducted experiments with validator sets of 6, 25, and 100 nodes, where all nodes functioned as validators. Block sizes of 1\,MB and 0.1\,MB were tested. Each experiment was executed for 1 to 1.5 hours.

\subsection{Failure-free Performance}

We collected various performance metrics during the experiments, including CPU and RAM usage, block rate, message processing times, and replication lag. The results are summarized below.

\begin{table}[h!]
    \centering
    \caption{Performance Metrics for Different Validator Configurations}
    \label{tab:performance_metrics}
    \resizebox{\textwidth}{!}{\begin{tabular}{|l|c|c|c|c|}
        \hline
        \textbf{Metric} & \textbf{6 Validators (1MB)} & \textbf{25 Validators (1MB)} & \textbf{100 Validators (1MB)} & \textbf{100 Validators (0.1MB)} \\ \hline
        \textbf{CPU Usage}       & 0.5 cores avg, 0.7 cores max & 1 core avg, 1.5 cores max & 2.2 cores avg, 3.5--4 cores max & 2.2 cores avg, 3.2--4 cores max \\ \hline
        \textbf{RAM Usage}       & 0.7\,GB avg, spikes to 4.5\,GB & 0.7\,GB avg, spikes to 5\,GB & 1\,GB avg, spikes to 6\,GB & 1\,GB avg, spikes to 6\,GB \\ \hline
        \textbf{Round-trip Latency} & 120\,ms avg              & 130\,ms avg              & 140\,ms avg               & 140\,ms avg               \\ \hline
        \textbf{Block Rate}      & 3.7 blocks/s (83\% opt.)    & 3.3 blocks/s (80\% opt.) & 2.1 blocks/s (52\% opt.)  & 3 blocks/s (80\% opt.)    \\ \hline
        \textbf{Leader Proposal} & 17\,ms/msg                & 20\,ms/msg               & 21\,ms/msg                & 7\,ms/msg                 \\ \hline
        \textbf{Replica Commit}  & 4\,ms/msg                 & 3\,ms/msg                & 2.6\,ms/msg               & 2.9\,ms/msg               \\ \hline
        \textbf{New View}        & 4--24\,ms/msg             & 11\,ms/msg               & 7--30\,ms/msg             & 9\,ms/msg                 \\ \hline
        \textbf{Replication Lag} & 0--2 blocks avg, 1--5 blocks max & 0--7 blocks avg, 1--10 blocks max & 1--7 blocks avg, 3--10 blocks max & 1--7 blocks avg, 3--10 blocks max \\ \hline
    \end{tabular}
    }
\end{table}

\textbf{Block Rate.}  Reducing the number of validators leads to a block rate of 3.7 blocks per second. Decreasing the block size to 0.1\,MB resulted in a block rate of 3 blocks per second. Optimizing these parameters is beyond the scope of this report. Our current BFT implementation achieves an optimal number of round-trips. The next bottleneck worth addressing is the \textit{proposal broadcast bandwidth}; currently, the leader sends a separate copy of the proposal to each replica.

\textbf{Message Processing Time.} We observed instability in the new view processing time in both the 6-validator and 100-validator setups. The cause of this instability is currently unknown. Notably, almost all new view messages included commit certificates rather than timeout certificates, which should have ensured a consistent workload. Additionally, processing duplicate new view messages resulted in errors, and such duplicates were excluded from the metrics. As a result, they should not have influenced the statistics.

\textbf{CPU Usage.} The experiments demonstrate that the consensus protocol scales effectively with the number of validators and maintains reasonable resource utilization. The block rate and CPU usage are influenced by the number of validators and block size, highlighting areas for potential optimization, such as improving proposal broadcast efficiency. Future work will focus on addressing observed instabilities in message processing times and further optimizing resource utilization.

\section{Conclusions}

ChonkyBFT represents a significant advancement in Byzantine fault-tolerant consensus protocols, tailored to meet the practical demands of modern blockchain systems like ZKsync. By prioritizing simplicity, single-slot finality, and reduced systemic complexity, the protocol achieves a balance between theoretical rigor and real-world applicability. Its hybrid design, drawing from FaB Paxos, Fast-HotStuff, and HotStuff-2, enables low user-perceived latency — a critical requirement for high-performance rollups.

Key properties include:
\begin{itemize}
    \item \textbf{Quadratic communication with $n \ge 5f+1$ fault tolerance.} The protocol achieves safety under adversarial conditions with $n \ge 5f+1$ replicas, ensuring tolerance for up to $f$ Byzantine faults. This trade-off enables single voting-round finality while maintaining quadratic message complexity, simplifying implementation, and analysis.
    \item \textbf{Re-proposal mechanisms.} To prevent rogue blocks and ensure safety under adversarial conditions, ChonkyBFT mandates re-proposing previous blocks if a timeout occurs where a commit quorum certificate might have formed. This guarantees chain consistency even during network partitions or malicious leader behavior.
    \item \textbf{TimeoutQC-driven view changes.} Liveness is preserved through timeout QCs, ensuring deterministic progress without relying on timing assumptions. This avoids subtle vulnerabilities seen in protocols like Tendermint while maintaining safety.
     \item \textbf{Formal verification.} Safety properties were rigorously validated using Quint and Apalache, exposing edge cases and ensuring protocol correctness. This approach enhances trust in the protocol and provides a blueprint for integrating formal methods into consensus design.
\end{itemize}

Empirical evaluations demonstrated ChonkyBFT's scalability, with throughput remaining stable across configurations (6–100 validators) and block sizes (0.1–1 MB). The protocol achieved 3.7 blocks/second with 6 validators and maintained 2.1 blocks/second at 100 nodes, showing quadratic communication's practicality in global networks for moderate numbers of validators.

The protocol's safety proofs and model-checked invariants provide strong assurances for deployment, addressing historical challenges in BFT consensus verification. By decoupling theoretical metrics from practical performance considerations, ChonkyBFT offers a blueprint for consensus designs that value implementability alongside algorithmic elegance. Future work includes enhancing proposal dissemination mechanisms, improving fault tolerance to $n \ge 5f-1$~\cite{NoteFAB}, and further formalization in proof assistants like Isabelle, Coq, or Lean.

\section*{Acknowledgments}

We thank Denis Firsov for his valuable discussions on the protocol specification and 
verification efforts. In particular, he suggested introducing separate lists for quorum 
certificates, which was crucial in enabling the efficient checking of the inductive 
invariant.

\bibliographystyle{unsrt}
\bibliography{bibliography}

\begin{thebibliography}{10}

\bibitem{Buterin2022Complexity}
Vitalik Buterin.
\newblock Encapsulated vs systemic complexity in protocol design.
\newblock \url{https://vitalik.eth.limo/general/2022/02/28/complexity.html},
  Feb 2022.
\newblock Accessed: 2024-12-07.

\bibitem{2roundBFT}
Decentralized Thoughts.
\newblock 2-round bft smr with $n=4$, $f=1$.
\newblock
  \url{https://decentralizedthoughts.github.io/2021-03-03-2-round-bft-smr-with-n-equals-4-f-equals-1/},
  2021.
\newblock Accessed: 2024-12-07.

\bibitem{Gelashvili2021JolteonDitto}
Rati Gelashvili, Lefteris Kokoris-Kogias, Alberto Sonnino, Alexander
  Spiegelman, and Zhuolun Xiang.
\newblock Jolteon and ditto: Network-adaptive efficient consensus with
  asynchronous fallback.
\newblock {\em arXiv preprint arXiv:2106.10362}, 2021.
\newblock Available at: \url{https://arxiv.org/abs/2106.10362}.

\bibitem{Dai2023ParBFT}
Xiaohai Dai, Bolin Zhang, Hai Jin, and Ling Ren.
\newblock {ParBFT}: Faster asynchronous {BFT} consensus with a parallel
  optimistic path.
\newblock In {\em Proceedings of the 2023 {ACM} {SIGSAC} Conference on Computer
  and Communications Security}, pages 1--15, 2023.
\newblock Available at: \url{https://eprint.iacr.org/2023/679}.

\bibitem{dwork1988consensus}
Cynthia Dwork, Nancy~A. Lynch, and Larry~J. Stockmeyer.
\newblock Consensus in the presence of partial synchrony.
\newblock {\em Journal of the ACM}, 35(2):288--323, 1988.

\bibitem{bls_signature}
Dan Boneh, Sergey Gorbunov, Riad~S. Wahby, Hoeteck Wee, Christopher~A. Wood,
  and Zhenfei Zhang.
\newblock {BLS Signature}.
\newblock
  \url{https://datatracker.ietf.org/doc/draft-irtf-cfrg-bls-signature/},
  December 2022.
\newblock Internet-Draft, IRTF.

\bibitem{malkhi23}
Dahlia Malkhi and Kartik Nayak.
\newblock Extended abstract: Hotstuff-2: Optimal two-phase responsive {BFT}.
\newblock {\em {IACR} Cryptol. ePrint Arch.}, page 397, 2023.

\bibitem{ChonkyInformal}
Bruno França and Grzegorz Prusak.
\newblock {ChonkyBFT} informal specification.
\newblock
  \url{https://github.com/matter-labs/era-consensus/tree/main/spec/informal-spec},
  2024.
\newblock Accessed: 2024-12-19.

\bibitem{Quint}
Quint specification language.
\newblock \url{https://quint-lang.org/}, 2024.
\newblock Accessed: 2024-12-19.

\bibitem{Apalache}
Apalache: symbolic model checker for {TLA}$^+$ and {Quint}.
\newblock \url{https://apalache-mc.org/}, 2024.
\newblock Accessed: 2024-12-19.

\bibitem{KKT19}
Igor Konnov, Jure Kukovec, and Thanh{-}Hai Tran.
\newblock {TLA+} model checking made symbolic.
\newblock {\em Proc. {ACM} Program. Lang.}, 3({OOPSLA}):123:1--123:30, 2019.

\bibitem{KonnovKM22}
Igor Konnov, Markus Kuppe, and Stephan Merz.
\newblock Specification and verification with the {TLA}\({}^{\mbox{+}}\)
  trifecta: {TLC}, {Apalache}, and {TLAPS}.
\newblock In {\em ISoLA}, volume 13701 of {\em LNCS}, pages 88--105. Springer,
  2022.

\bibitem{Newcombe14}
Chris Newcombe.
\newblock Why {Amazon} chose {TLA} +.
\newblock In Yamine~A{\"{\i}}t Ameur and Klaus{-}Dieter Schewe, editors, {\em
  ABZ}, volume 8477, pages 25--39. Springer, 2014.

\bibitem{ChonkyBlogpost}
Denis Firsov, Bruno França, Igor Konnov, Denis Kolegov, and Grzegorz Prusak.
\newblock Specification and model checking of {BFT} consensus by {Matter Labs}.
\newblock
  \url{https://protocols-made-fun.com/consensus/matterlabs/quint/specification/modelchecking/2024/07/29/chonkybft.html},
  2024.
\newblock Accessed: 2024-12-19.

\bibitem{Lamport2002}
Leslie Lamport.
\newblock {\em Specifying Systems, The {TLA+} Language and Tools for Hardware
  and Software Engineers}.
\newblock Addison-Wesley, 2002.

\bibitem{tla-paxos}
Leslie Lamport.
\newblock {TLA}$^+$ specification of {Paxos}.
\newblock
  \url{https://github.com/tlaplus/Examples/blob/master/specifications/Paxos/Paxos.tla},
  2024.
\newblock Accessed: 2024-12-17.

\bibitem{tla-raft}
Diego Ongaro.
\newblock {TLA}$^+$ specification of {Raft}.
\newblock \url{https://github.com/ongardie/raft.tla}, 2024.
\newblock Accessed: 2024-12-17.

\bibitem{TendermintSpec2020}
Igor Konnov and Zarko Milosevic.
\newblock {TLA}$^+$ specification of {Tendermint} consensus and its
  accountability.
\newblock
  \url{https://github.com/cometbft/cometbft/blob/main/spec/light-client/accountability/TendermintAcc\_004\_draft.tla},
  2024.
\newblock Accessed: 2024-10-23.

\bibitem{tla-tetrabft}
Giuliano Losa.
\newblock {TLA}$^+$ specification of {TetraBFT}.
\newblock \url{https://github.com/nano-o/tetrabft-tla}, 2024.
\newblock Accessed: 2024-12-17.

\bibitem{tla-examples}
A collection of {TLA}$^+$ specifications of varying complexities.
\newblock \url{https://github.com/tlaplus/Examples}, 2024.
\newblock Accessed: 2024-12-17.

\bibitem{pluscal}
Leslie Lamport.
\newblock The {PlusCal} algorithm language.
\newblock \url{https://lamport.azurewebsites.net/pubs/pluscal.pdf}, 2017.

\bibitem{ChonkyQuint}
Igor Konnov and Denis Kolegov.
\newblock {ChonkyBFT} specification in {Quint}.
\newblock
  \url{https://github.com/matter-labs/era-consensus/tree/main/spec/protocol-spec},
  2024.
\newblock Accessed: 2024-12-19.

\bibitem{tange_2024_11247979}
Ole Tange.
\newblock Gnu parallel 20240522 ('tbilisi'), May 2023.
\newblock {GNU Parallel is a general parallelizer to run multiple serial
  command line programs in parallel without changing them.}

\bibitem{jackson2012software}
Daniel Jackson.
\newblock {\em Software Abstractions: logic, language, and analysis}.
\newblock MIT press, 2012.

\bibitem{NoteFAB}
Ittai Abraham, Kartik Nayak, Ling Ren, and Zhuolun Xiang.
\newblock Brief note: Fast authenticated byzantine consensus.
\newblock {\em CoRR}, abs/2102.07932, 2021.

\end{thebibliography}

\end{document}